\documentclass[11pt]{article}
\usepackage{amsmath, amssymb, amsfonts, amsthm}

\numberwithin{equation}{section}

\newtheorem{theorem}{Theorem}[section]
\newtheorem{proposition}[theorem]{Proposition}
\newtheorem{corollary}[theorem]{Corollary}
\newtheorem{lemma}[theorem]{Lemma}

\newcommand{\tri}{| \! | \! |}
\newcommand{\rd}{{\rm d}}
\newcommand{\be}{\begin{equation}}
\newcommand{\ee}{\end{equation}}
\newcommand{\bey}{\begin{eqnarray}}
\newcommand{\eey}{\end{eqnarray}}

\newcommand{\eps}{\varepsilon}

\newcommand{\bx}{{\bf x}}

\newcommand{\ph}{\varphi}

\newcommand{\om}{{\omega}}

\newcommand{\bR}{{\mathbb R}}

\newcommand{\bN}{{\mathbb N}}

\newcommand{\wt}{\widetilde}

\newcommand{\cG}{{\cal G}}

\newcommand{\cF}{{\cal F}}

\newcommand{\cE}{{\cal E}}

\newcommand{\cK}{{\cal K}}

\newcommand{\cL}{{\cal L}}

\newcommand{\cR}{{\cal R}}
\newcommand{\cJ}{{\cal J}}

\newcommand{\supp}{\operatorname{supp}}

\input epsf

\newcommand{\fh}{{\mathfrak h}}

\newcommand{\donothing}[1]{}

\newcommand{\subA}{_{\!_{_{A\!\!}}}}

\newcommand{\inthree}{\int_{\mathbb{R}^3}}
\newcommand{\inthreeN}{\int_{\mathbb{R}^{3N}}}

\usepackage{bbm}

\begin{document}



\title{Dynamical formation of correlations in a Bose-Einstein condensate}
\author{L\'aszl\'o Erd\H os\thanks{Partially supported by the SFB-TR12 of
the German Science Foundation}\:,
Alessandro Michelangeli\thanks{Supported by a Kovalevskaja Award from the
Humboldt Foundation.}\:,
and Benjamin Schlein$^{\dagger}$\thanks{On leave from Cambridge University}\\
\\
Institute of Mathematics, University of Munich, \\
Theresienstr.~39, D-80333 Munich, Germany
}

\maketitle
\begin{abstract}
We consider the evolution of $N$ bosons interacting
with a repulsive short range pair potential
 in three dimensions. The potential is scaled according
to the Gross-Pitaevskii scaling, i.e. it is given by  $N^2V(N(x_i-x_j))$.
We monitor the behavior of the solution to the $N$-particle
Schr\"odinger equation in a spatial window where two particles
 are close to each other. We prove that within this
window a short scale interparticle structure emerges dynamically.
The local correlation between the particles is given by the
two-body zero energy scattering mode. This is the characteristic
structure that was expected to form within a very short initial time
layer and to persist for all later times, on the basis of the
validity of the Gross-Pitaevskii equation for the evolution
of the Bose-Einstein condensate. The zero energy scattering
mode emerges after an initial time layer where all higher energy
modes disperse out of the spatial window. We can prove the persistence
of this structure up to sufficiently small times
before three-particle correlations could develop.

\bigskip


\noindent\textbf{AMS classification numbers:} 81U05,  81U30, 81V70, 82C10.
\end{abstract}

\section{Introduction and main result}

We consider a three-dimensional system of $N$ indistinguishable
spinless bosons coupled with a pairwise repulsive interaction $V_N$.
The Hamiltonian of this system is
\begin{equation}\label{HN}
H_N=\sum_{i=1}^N(-\Delta_i)+\!\!\sum_{1\leqslant i<j\leqslant N}V_N(x_i-x_j)
\end{equation}
acting on $L^2_{\mathrm{sym}}(\mathbb{R}^{3N}, \rd \bx)$, the symmetric
subspace of the tensor product of $N$ copies of the one-particle space $L^2(\bR^3)$. Here $\bx =(x_1, x_2, \ldots, x_N)\in \bR^{3N}$ denotes the position of the $N$ particles.
We assume that the interaction potential $V_N$ scales with $N$ according to
the  scaling introduced by Lieb, Seiringer and Yngvason in \cite{LSY}; that is,
we fix a repulsive potential $V$ and we rescale it by defining
\begin{equation}
V_N(x) \; :=\;N^2V(Nx)\,.
\end{equation}
We assume that the unscaled potential $V:\mathbb{R}^3\to\mathbb{R}$ is non-negative, smooth, spherically symmetric, and compactly supported. The wave function of the system at time $t$ is denoted by $\Psi_{N,t}\in L^2_{\mathrm{sym}}(\mathbb{R}^{3N})$ with $\|\Psi_{N,t}\|_2=1$. It evolves
according to the  Schr\"odinger equation
\begin{equation}\label{HNev}
i\partial_t\Psi_{N,t} \; = \; H_N\Psi_{N,t}
\end{equation}
with a given initial condition $\Psi_{N,t=0}=\Psi_N$. We will be interested in the evolution of initial data exhibiting complete Bose-Einstein condensation.

\medskip

For a given wave function $\Psi_N\in L^2_{\mathrm{sym}}(\mathbb{R}^{3N})$
we define the one-particle marginal $\gamma^{(1)}_{\Psi_N}= \gamma^{(1)}_N$
associated with $\psi_N$
 to be the positive trace-class operator on $L^2(\mathbb{R}^{3k})$ with kernel given by
\begin{equation}\label{eq:1part}
\gamma_N^{(1)}(x ; x') \; := \; \int_{\mathbb{R}^{3(N-1)}}\Psi_N(x,z_2,\dots,z_N)
\overline{\Psi_N(x',z_2,\dots,z_N)}\,\rd z_2\cdots\rd z_N\,.
\end{equation}
We say that a sequence $\{\Psi_N\}_{N\in\mathbb{N}}$ of $N$-body wave functions
exhibits {\it complete Bose-Einstein condensation} in the one-particle state  $\varphi\in L^2({\mathbb{R}^3})$, $\|\varphi\|_2=1$, if
\begin{equation}\label{defbec}
\gamma_{N}^{(1)}\longrightarrow|\varphi\rangle\langle\varphi|
\end{equation}
in the trace-norm topology, as $N\to\infty$. Here $|\varphi\rangle\langle\varphi|$ denotes the orthogonal projection operator onto $\varphi$ (Dirac notation).

\medskip

On the level of the one-particle density
matrix, the condition \eqref{defbec} is a signature that almost all particles (up to a fraction vanishing in the limit $N \to \infty$) occupy the same one-particle state, described by the orbital $\ph$ (condensate wave function). Suppose now that a trapping external potential is added to the Hamiltonian, i.e., instead of \eqref{HN} consider the Hamiltonian
\be
   H_N^{\text{trap}} \; = \;  \sum_{i=1}^N(-\Delta_i+U(x_i))
+\!\!\sum_{1\leqslant i<j\leqslant N}V_N(x_i-x_j)
\label{trappedH}
\ee
with $U(x)\to \infty$ as $|x|\to\infty$. In \cite{LSY}, Lieb, Seiringer and Yngvason proved that the ground state energy of $H_N^{\text{trap}}$ divided by the number of particle $N$ (the ground state energy per particle) converges, as $N \to \infty$, to the minimum of the Gross-Pitaevskii energy functional
\be\label{GPF}
\cE (\varphi) =  \int_{\bR^3} \left( |\nabla\varphi|^2 + U|\varphi|^2 + 4\pi a |\varphi|^4\right)
\ee
over all $\ph\in L^2 (\bR^3)$ with $\|\ph \|_2=1$. The coupling constant $a$
is the scattering length of the unscaled potential $V$. Later, in \cite{LS}, Lieb and Seiringer also showed that the ground state of $H_N^{\text{trap}}$ exhibits complete Bose-Einstein condensation into the minimizer of (\ref{GPF}).

\medskip

We recall that the scattering length of a potential $V$ is defined as
\begin{equation}\label{scattlength}
a \; := \; \lim_{|x|\to\infty} |x|\,\omega(x) \; = \; \frac{1}{8\pi}
\inthree\rd x V(x)\big(1-\omega(x)\big)\,,
\end{equation}
where $1-\omega$ is the unique non-negative solution of the zero-energy scattering equation
\begin{equation}\label{eq:scatteq}
\big(\!-\Delta+\!\!\begin{array}{l}\frac{1}{2}\end{array}\!\!V\big)(1-\omega) \; = \; 0
\end{equation}
with boundary conditions $\omega(x)\to 0$ as $|x|\to\infty$.
Several properties of $\omega$ are collected in Appendix~\ref{app1bodyscatt}.
It follows from the definition that the scattering length of the scaled potential $V_N$ is $a_N=a/N$ and that the zero-energy scattering solution associated with $-\Delta+\frac{1}{2}V_N$ is the function $1-\omega_N$ with $\omega_N(x)=\omega(Nx)$. In particular, $\omega_N$ has a built-in structure at the scale $N^{-1}$.

\medskip

The time evolution of a condensate after removing the trap, i.e. setting  $U\equiv 0$,
can be described by the solution of the Gross-Pitaevskii equation
\begin{equation}\label{tGP}
i\partial_t\varphi_t \;=\;-\Delta\varphi_t\;+\;8\pi a|\varphi_t|^2\varphi_t
\end{equation}
with initial condition $\varphi_t\big|_{t=0}=\varphi$.

\medskip

The first rigorous derivation of the Gross-Pitaevskii equation \eqref{tGP} from
many-body Schr\"odinger dynamics \eqref{HNev} has been obtained in \cite{ESY-2006} under the condition that the unscaled interaction potential $V$ is sufficiently small. This result
was later extended to large interaction potentials in \cite{ESY-2008}. More precisely, it has been shown that if the family of wave functions $\{ \Psi_N \}_{N \in \bN}$ has finite energy per particle (in the sense that $\langle\Psi_{N},H_N\Psi_{N}\rangle\leqslant CN$), and if it exhibits complete Bose-Einstein condensation in the one-particle state $\varphi$ in the sense (\ref{defbec}), then, at any time $t\neq 0$, $\Psi_{N,t}=e^{-iH_Nt}\Psi_N$ still exhibits complete Bose-Einstein condensation. Moreover, the condensate wave function  $\varphi_t$ at later times is determined by the solution of the Gross-Pitaevskii equation \eqref{tGP} with initial data $\varphi_{t=0}=\varphi$.

\medskip

At first sight, the condition \eqref{defbec} may indicate that wave functions describing  condensates are very close to be factorized (for $\Psi_N = \varphi^{\otimes N}$, \eqref{defbec} holds as an equality). However, the structure of the evolved wave function $\Psi_{N,t}$ is much more complicated. In fact, it turns out that $\Psi_{N,t}$ is characterized by a short scale correlation structure which plays an important role in the derivation of the Gross-Pitaevskii equation. If this short-scale structure were lacking along the evolution and the wave function $\Psi_{N,t}$ were essentially constant in the relative coordinates $x_i-x_j$ whenever $|x_i-x_j|\lesssim \frac{1}{N}$, then the condensate wave function $\varphi_t$ would solve the equation \eqref{tGP} with coupling constant $b$ (this is the first Born approximation  of $8\pi a$); we discuss this problem in more details in Appendix \ref{app:GPa}.

\medskip



If the initial state $\Psi_N$ has a short-scale structure
characterized by $1-\omega_N(x_i-x_j)$ for $|x_i-x_j|\lesssim \frac{1}{N}$,
for example, it is the ground state of the system before removing
the traps, then this structure seems to {\it persist} along the
time evolution. The condensate wave function $\varphi_t$ lives
on order one scales and it changes with time, while the correlation
structure on the scale $N^{-1}$ is conserved. However, Theorem 2.2
of \cite{ESY-2006} states that even if the initial state
does not have a built-in short-scale structure, the time evolution
of the orbital wave function is still given by \eqref{tGP}.
This indicates that the characteristic short-scale
structure not only persists but {\it dynamically emerges}
within a very short initial time layer.

\medskip

The main result of this paper is a description of the dynamical
formation of this short-scale structure. Consistently with the underlying Gross-Pitaevskii equation, which is derived at any fixed $t$ in the limit $N\to\infty$, correlations
in a large but finite system with $N$ particles are expected to form in a short transient
 time (i.e., $o(1)$ in $N$) and then to be preserved
 at any macroscopic time (i.e., times of order 1).

\medskip

In the previous works the
 presence of the local structure was  captured by
the $H_N^2$-energy estimate
\begin{equation}\label{H2}
\inthreeN\Big|\,\nabla_i\nabla_j\frac{\Psi_N(\bx)}{\,1-\omega_N(x_i-x_j)}\,\Big|^2\,
\rd\bx\;\leqslant\;C\,\Big\langle\Psi_N,\frac{H_N^2}{\:N^2\,}\Psi_N\Big\rangle
\end{equation}
valid for every $\Psi_N\in L^2(\mathbb{R}^{3N})$ and every fixed
indices $i\neq j$ in $\{1,\dots N\}$. This inequality has been proven in \cite{ESY-2006} for small interaction potential (for large interaction potential, another a-priori bound on the solution $\Psi_{N,t}$ of the $N$-particle Schr\"odinger equation has been obtained in \cite{ESY-2008}).

\medskip

Under very general conditions on the trapping
potential $U$,
the r.h.s of  (\ref{H2}) is of order one for the ground state
$\Psi_N$ of $H_N^{\text{trap}}$.
If $\Psi_N$ did not carry a short-scale structure $(1-\omega_N(x_i-x_j))$ when
 $|x_i-x_j|\sim N^{-1}$, say, it were essentially constant
on this scale, then the main term on the l.h.s of (\ref{H2}) would be  of
order
$$
     \int_{|x_i-x_j|\leqslant  1/N}
 \Big|\,\nabla_i\nabla_j\frac{1}{\,1-\omega_N(x_i-x_j)}\,\Big|^2
    |\Psi_N(\bx)|^2 \rd \bx \sim
 \int_{|x|\leqslant  1/N} |\nabla^2\om_N(x)|^2 \rd x\sim O(N)\; .
$$
More generally, this shows  that the short-distance behaviour of any $\Psi_N$
for which the inequality
\be
   \langle \Psi_N, H_N^2\Psi_N\rangle \leqslant  CN^2
\label{H2bound}
\ee
holds,
 is asymptotically given
by $(1-\omega_N(x_i-x_j))$ in the regime where $|x_i-x_j|\sim N^{-1}$.
Since $H_N^2$ is conserved along the dynamics,
$$
    \langle\Psi_{N,t}, H_N^2 \Psi_{N,t}\rangle
 = \langle\Psi_N, H_N^2\Psi_N\rangle,
$$
the same conclusion holds for $\Psi_{N,t}$ at any later time $t>0$,
thus for such initial data the
 local structure is preserved in time.

\medskip

The same argument cannot be used to detect the presence and the formation of the
short-scale structure
for  states that do not satisfy the bound \eqref{H2bound}.
 This is the case for the completely uncorrelated initial state $\varphi^{\otimes N}$,
with some $\varphi\in H^2(\mathbb{R}^3)$, because
$\langle\varphi^{\otimes N} , H_N^2\varphi^{\otimes N}\rangle\sim N^3$
 (see Lemma~\ref{prodphi-N3}). However, the validity of the Gross-Pitaevskii
equation for the product initial state (Theorem~2.2 of~\cite{ESY-2006})
indicates that the short-scale structure, although initially
not present,  still forms within a very short time interval.
The length of this transient time interval must  vanish as $N\to\infty$
since after the $N\to\infty$ limit, the Gross-Pitaevskii equation \eqref{tGP} is
valid for any positive time $t>0$.
The proof of Theorem  2.2.~of \cite{ESY-2006} still relied on \eqref{H2bound}
after an energy cutoff $H_N\leqslant  \kappa N$ in the many-body Hilbert space. The energy cutoff
 artificially introduced the short-scale structure in the low energy regime.
 Although it was showed that, in the case $\kappa\gg 1$,
 the high energy regime $H_N\geqslant \kappa N$
essentially does not influence the evolution of the condensate, the proof
did not reveal whether the short-scale structure is indeed formed
in the full wave function without the cutoff.
The goal of the present analysis is to prove the dynamical formation of the expected pattern
 of correlations in the time evolution of an initially uncorrelated many-body state
 $\Psi_N=\varphi^{\otimes N}$.

\medskip

The above discussion suggests to compare $\Psi_{N,t}$ with
\begin{equation}\label{trial}
\prod_{1\leqslant i<j\leqslant N}\!\!\!\big(1-\omega_N(x_i-x_j)\big)\,\prod_{r=1}^N\varphi_t(x_r)
\end{equation}
after the transient time.
 Although (\ref{trial}) has the
expected built-in short-scale structure in each relative variable,
the true $N$-body wave function at later time $t>0$ is much more complicated
and it is beyond our reach to describe it precisely. The main reason
is that it quickly develops higher order correlations as well;
the typical distance between neighboring particles is of order $N^{-1/3}$,
so within time $t$ of order one, each particle collides with many of its neighbors.
Moreover, the energy of the product initial state,
$$
  \langle \varphi^{\otimes N}, H_N\varphi^{\otimes N}\rangle
\approx N \Big( \int_{\bR^3} |\nabla \varphi|^2
   + \frac{b}{2} \int_{\bR^3} |\varphi|^4\Big)
$$
is much bigger than the energy predicted by the Gross-Pitaevskii functional,
by recalling that $b> 8\pi a$. Thus there is an excess energy of order $N$
in the system that remains unaccounted for. This excess energy must live
on intermediate length scales   $N^{-\gamma}$ for $0< \gamma<1$ not to be
detected either on the local structure of order $N^{-1}$ or on the
order one scale of the condensate. Moreover, these excess modes
 must be sufficiently incoherent not to influence the evolution
of the condensate.

\medskip

We cannot describe the evolution of these intermediate modes yet, we
can only focus on the formation of the local structure $(1-\omega_N(x_i-x_j))$.
This is selected in a scattering process by having the smallest local energy.
Our study is restricted to a
spatial window  where $|x_i-x_j|<\ell$ for some intermediate $\ell$ between the
length scale $N^{-1}$ of the expected local structure and the macroscopic order one
scale of the system.
If the initial datum is essentially constant in the variable $x_i-x_j$
as far as $|x_i-x_j|<\ell$, then a short-scale
structure emerges corresponding to the zero-energy mode $1-\omega_N$ together with all higher
energy modes; these quickly
leave the $\ell$-window while only the short-scale structure remains.
We can thus monitor the formation of the local structure $(1-\omega_N(x_i-x_j))$
in this window and prove its persistence for a while after its emergence.
The time scale must be sufficiently large compared to the window size
so that all higher energy modes disperse, but it has to be sufficiently small
so that no three-particle correlations could develop yet.

\medskip

We provide a rigorous version of this picture in Theorem \ref{N-thm}.
We focus on the local structure in the relative variable $x_1-x_2$;
by symmetry the same result holds for any pairs of relative variables $x_i-x_j$.
We define a cutoff function
\begin{equation}\label{cutofftheta}
\theta_\ell(x):=\chi\Big(\frac{|x|}{\ell}\Big)
\end{equation}
with
\begin{equation}\label{def_cutoff}
\chi\in C^{\infty}(\mathbb{R}_+)\,,\qquad\chi(r) \;:=\;\begin{cases}
1 & \textrm{if }\; 0\leqslant r\leqslant 1 \\
0 & \textrm{if }\; r\geqslant 2\,.
        \end{cases}
\end{equation}
We consider the following time-dependent quantity
\begin{equation}\label{FN}
\cF_N(t)\;:=\;\inthreeN\theta_\ell(x_1-x_2)\,\bigg|\,
\frac{\,\Psi_{N,t}(\bx)\,}{1-\omega_N(x_1-x_2)}-\Psi_N (\bx) \,\bigg|^2\,\rd \bx
\end{equation}
where $1-\omega_N$ is the zero energy scattering solution of $-\Delta+\frac{1}{2}V_N$, and $\Psi_N (\bx) = \prod_{j=1}^N \ph (x_j)$ is the initial $N$-body wave function. The cutoff scale $\ell$ is always assumed to satisfy
\begin{equation}
\frac{1}{N} \; \leqslant \; \ell \; \ll \; \frac{1}{\sqrt{N\,}\,}\,.
\end{equation}
Here and throughout in the sequel we will make the convention that by $A_N\ll B_N$
one means that
$0<(\log N)^kA_N\leqslant  C B_N$ with a sufficiently large $k$ and $C$.

\medskip

The decrease of $\cF_N(t)$ has the natural interpretation of formation of a local structure
at the scale $N^{-1}$ within the $\ell$-window.
 By controlling this quantity  we demonstrate that the conjectured
short-scale structure indeed forms within a short time of order $N^{-2}$ and then it is
preserved for longer times in a time window that is essentially $N^{-2}\ll t
\ll N^{-(2-\frac{1}{10})}$.
 Our main result is the following:

\medskip

\begin{theorem}\label{N-thm}
Let $V:\mathbb{R}^3\to\mathbb{R}$ be a non-negative, smooth, spherically symmetric, and compactly supported potential. Let $V_N(x):=N^2V(Nx)$ for $N\in\mathbb{N}$. Let $\varphi\in L^2(\mathbb{R}^3)$ with $\|\varphi\|_2=1$ and
\begin{equation}\label{assumption-on-phi}
\|\varphi\|_{4,\infty,\alpha} \; :=
\; \sum_{m=0}^4\big\|\langle x \rangle^\alpha \nabla^m \varphi\big\|_{\infty} \; < \; \infty\,
\end{equation}
for some $\alpha >3$. Consider the Hamiltonian
\begin{equation}\label{N-Hamiltonian}
H_N=\sum_{i=1}^N(-\Delta_i)\;+\!\!\sum_{1\leqslant i<j\leqslant N}\!\!V_N(x_i-x_j)
\end{equation}
acting on $L^2(\mathbb{R}^{3N})$, the initial datum $\Psi_N=\varphi^{\otimes N}$, and its time evolution $\Psi_{N,t}:=e^{-iH_Nt}\varphi^{\otimes N}$. Consider the function $\cF_N(t)$ defined in (\ref{FN}). Then
\begin{equation}\label{final_Nbody_estimate}
\cF_N(t) \; \leqslant \; C \, \cF_N(0) \bigg(\:  \frac{(\log N)^{\frac{4}{5}}}{N^{\frac{1}{5}}}\frac{(N^2t)^2}{N\ell}+ \, \frac{\:(N\ell)^4}{N^2t} \left(\log N^2t \right)^6 \bigg)
\end{equation}
for all times $t$ such that $0 < t \ll N^{-1}$. As a consequence,
\begin{equation}\label{t-range}
\cF_N(t) \ll \, \cF_N(0) \qquad\textrm{for}\qquad
(N\ell)^4 \; \ll \; N^2t \; \ll \; N^{\frac{1}{10}}(N\ell)^{\frac{1}{2}}\,.
\end{equation}
\end{theorem}

\medskip

\emph{Remark 1}. Eq. (\ref{final_Nbody_estimate}) does not look dimensionally correct (in our coordinates, time has the dimension of a length squared). The reason is that, to simplify the notation, we consider the length scale $\lambda$ characterizing the initial wave function $\ph$, the radius $R$ of the support of the (unscaled) potential $V$, and the scattering length $a$ of $V$ as dimensionless constants of order one. More generally, if, for $\lambda >0$ we set \[ \ph^{(\lambda)} (x) = \lambda^{-3/2} \ph (x/\lambda) \] the bound Eq. (\ref{final_Nbody_estimate}) assumes the (dimensionally correct) form
\[ \cF_N(t) \; \leqslant \; C \, \cF_N(0) \bigg(\: \frac{1}{\lambda^3} \frac{(\log N)^{\frac{4}{5}}}{N^{\frac{1}{5}}}\frac{(N^2t)^2}{N\ell} + \,\frac{1}{\lambda^2} \frac{\:(N\ell)^4}{N^2t} \left(\log (N^2t/\lambda^2) \right)^6 \bigg) \] where the dimensionless constant $C$ depends on the ratios $R/\lambda$ and $a/\lambda$ (a more refined analysis would also allow to compute the precise dependence on $R$ and $a$).

\medskip

\emph{Remark 2}.
Since the interaction potential has a length scale
$1/N$, it will be convenient to introduce the length
$L=N\ell$, expressing the size of the window relative to
the interaction range. The two particle scattering process
takes place on a time scale $1/N^2$ (see Section \ref{sec:proof}),
thus it is also natural to introduce the rescaled time $T=N^2t$.
The appearance of the combinations of $N\ell$ and $N^2t$
in the theorem is motivated by the fact that we actually
describe a long time  and large distance scattering
process  in terms of the rescaled variables
$T\gg 1$, $L\geqslant 1$, in the regime
where $L^4 \ll T \ll N^{1/10} L^{1/2}$ (in the spirit of Remark 1 above, we consider the regime $(L/\lambda)^4 \ll (T/\lambda^2) \ll N^{1/10} (L/\lambda)^{1/2}$).

\medskip

\emph{Remark 3}. To be concrete, choosing $\ell=\frac{1}{N}$,
the short-scale structure given by $1-\omega_N$ is established for times
\begin{equation}
\frac{1}{N^2} \; \ll \; t \; \ll \; \frac{1}{N^{\,2-\frac{1}{10}}}\;.
\end{equation}
The first inequality corresponds to the formation of the short-scale structure
beyond  the scattering time scale within a
window comparable with the interaction range.
The second inequality expresses the persistence of this short-scale structure
within this spatial window up to times much longer than the scattering time.

\medskip

\emph{Remark 4}. Note that in the definition of $\cF_N(t)$
we compared $\Psi_{N,t}/(1-\omega_N)$ with the initial product
state and not with the evolved product state $\varphi_t^{\otimes N}$ which
would have been more natural. For the relatively short time scales that
we can consider, this difference is irrelevant; the condensate essentially does not move.
One can directly check that $\int \theta_\ell(x_1-x_2)|\varphi^{\otimes N}(\bx)
-\varphi_t^{\otimes N}(\bx)|^2\rd \bx \ll \cF_N(0)$, hence the modification
does not influence \eqref{t-range}.
To investigate the persistence of the local structure up
to times of order 1 in $N$, the definition of $\cF_N(t)$ should, of course,
contain $\varphi_t^{\otimes N}$
instead of the initial state $\varphi^{\otimes N}$.
However, at such large times, the $i$-th and $j$-th particles also interact with other particles.
Our analysis does not control consecutive multiple collisions, although the validity
of the Gross Pitaevskii equation still gives an indirect evidence that the correlation structure is
preserved even for times of order one.

\medskip

\emph{Remark 5}. In \cite{ESY-2006} the local structure was identified
by the $L^2$ norm of the mixed derivative $\nabla_1\nabla_2 [\Psi_{N,t}/(1-\omega_N)]$.
For initially factorized states, the integral
\begin{equation}\label{eq:rem4}
   \int_{\bR^{3N}} \theta_{\ell} (x_1 -x_2) \; \Big| \nabla_1\nabla_2
\frac{\Psi_{N,t}(\bx)}{1-\omega_N(x_1-x_2)}\Big|^2\rd\bx
\end{equation}
is of the order $N$ at time $t=0$, and it is expected to be of order one for times $t \gg N^{-2}$. Proving this decay would establish the formation of the local structure in a much stronger norm than the local $L^2$ norm used Theorem \ref{N-thm}. Unfortunately, due to the very singular interaction potential $V_N$, we cannot bound (\ref{eq:rem4}) effectively (estimating it by the expectation of $H_N^2$ produces a bound proportional to $N$ for all times, see Lemma \ref{prodphi-N3}).

\medskip

\emph{Remark 6}. Condition (\ref{assumption-on-phi}) encodes all
regularity and  decay that we assume on $\varphi$, although it is
not optimal and we do not aim at finding most general conditions
on $\varphi$. In particular, (\ref{assumption-on-phi}) implies
that $\varphi\in H^3(\mathbb{R}^3)$,  hence $\varphi\in C^1(\mathbb{R}^3)$.

\medskip

\emph{Notation}. By $C$ we will mean a constant depending only on the unscaled potential $V$ and the initial one-body wave function $\varphi$. Constants denoted by $c_{p,q,\dots}$ are meant to depend also on the indices $p,q$, etc.

\section{Proof of main Theorem}\label{sec:proof}

In this Section we present the main steps of the proof of Theorem \ref{N-thm}.

\medskip

\begin{proof}[Proof of Theorem \ref{N-thm}]
Let $\cF_N (t)$ be the quantity defined in (\ref{FN}).

\medskip

To evaluate $\cF_N$, we introduce a dynamics where particles $1$ and $2$ are decoupled from the others. We define
\begin{eqnarray}
\fh_N^{(1,2)} & := & -\Delta_1-\Delta_2+ V_N(x_1-x_2) \label{hN12} \\
H_N^{(3)} & := & \sum_{i=3}^N(-\Delta_i)+\sum_{3\leqslant i<j\leqslant N}V_N(x_i-x_j) \label{HN3}\\
U_N^{(1,2)} & : = & \sum_{j=3}^N\big(V_N(x_1-x_j)+V_N(x_2-x_j)\big) \\
\widetilde H_N & : = & \fh_N^{(1,2)} + H_N^{(3)} \;=\; H_N - U_N^{(1,2)} \label{Htilde}\,.
\end{eqnarray}
Then we have
\begin{equation}\label{Psi_Ntilde}
\widetilde\Psi_{N,t}\;:=\;e^{-i\widetilde H_N t}\varphi^{\otimes N}\;=\;\left(\,e^{-i \fh_N^{(1,2)} t}\varphi^{\otimes 2}\,\right) \otimes \left(\,e^{-i H_N^{(3)} t}\,\varphi^{\otimes (N-2)}\right)\;= \;\psi_{t}\otimes\Phi_t\,.
\end{equation}
with $\psi_t = e^{-i\fh_N^{(1,2)} t} \ph^{\otimes 2}$ and $\Phi_t = e^{-iH_N^{(3)} t} \ph^{\otimes (N-2)}$. Thus,
\begin{equation}\label{firstest_Ft}
\begin{split}
\cF_N(t)\; & \leqslant\; \inthreeN\theta_\ell(x_1-x_2)\,\bigg|\,\frac{\,\Psi_{N,t}(\bx)-\widetilde\Psi_{N,t}(\bx)\,}{1-\omega_N(x_1-x_2)}\,\bigg|^2\rd \bx\
 \\
& \qquad\qquad +\; \inthreeN\theta_\ell(x_1-x_2)\,\bigg|\,\frac{\,\widetilde\Psi_{N,t}(\bx)\,}{1-\omega_N(x_1-x_2)}-\prod_{i=1}^N\varphi(x_i)\,\bigg|^2\rd \bx\ \\
& \leqslant \; C \, \big(\, \cG_N(t) + \cK_N(t)\big)
\end{split}
\end{equation}
where
\begin{eqnarray}
\cG_N(t) & := & \inthreeN\theta^2_{2\ell}(x_1-x_2)\big|\Psi_{N,t}(\bx)-\widetilde\Psi_{N,t}(\bx)\big|^2\rd\bx\,, \label{gN}\\
\cK_N(t) & := & \inthreeN\theta_\ell(x_1-x_2)\,\bigg|\,\frac{\,\widetilde\Psi_{N,t}(\bx)\,}{1-\omega_N(x_1-x_2)}-\prod_{i=1}^N\varphi(x_i)\,\bigg|^2\rd \bx \label{kN}
\end{eqnarray}
(notice that $\theta_{\ell}(x)<\theta^2_{2\ell}(x)$ has been used).

\medskip

$\cG_N$ measures the $L^2$-distance, in the spatial region where $|x_1-x_2|\lesssim\ell$, between the evolutions $\Psi_{N,t}$ and $\widetilde\Psi_{N,t}$ of the initial state $\varphi^{\otimes N}$ with the dynamics given by $H_N$ and $\widetilde H_N$, respectively. By construction, $\cG_N(0)=0$ and at later times $\cG_N(t)$ deteriorates as the two vectors $\Psi_{N,t}$ and $\widetilde\Psi_{N,t}$ separate. We control such behavior in Section \ref{sec:Nbodyproblem}: the result (Proposition \ref{prop:GN}) is
\begin{equation}\label{bound_on_GN}
\cG_N(t) \; \leqslant \; C\,(N\log N)^{\frac{4}{5}}\,t^{2}
\end{equation}
provided that $\ell\ll N^{-2/5}$ and for all times $t\geqslant 0$. Notice that, if we had followed the dependence on the length scale $\lambda$ characterizing the initial wave function $\ph$ (as explained in Remark 1 after Theorem \ref{N-thm}), the bound (\ref{bound_on_GN}) would have taken the dimensionally correct form $\cG_N (t) \leq C (N \log N)^{4/5} (t/\lambda^2)^2$; similar remarks would also apply to all estimates in the sequel which are apparently dimensionally incorrect.

\medskip

On the other hand, since factorisation $\widetilde\Psi_{N,t}=\psi_{t}\otimes\Phi_t$ is preserved in time, see (\ref{Psi_Ntilde}), then $\cK_N$ turns out to be essentially a two-body quantity, that is, an integral only in variables $x_1$ and $x_2$, which involves the Schr\"odinger evolution of the initial datum $\varphi^{\otimes 2}$ with interaction $V_N(x_1-x_2)$. Reduction of $\cK_N$ to a two-body integral and a remainder is done in Section \ref{sec:2bodyproblem}. The result (Proposition \ref{prop:reduction}) is
\begin{equation}\label{bound_on_KN}
\cK_N(t) \; \leqslant \; C\ell^2\!\int_{\mathbb{R}^3\times\mathbb{R}^3}\!\rd \eta\,\rd x\,\theta_{2\ell}(x)\bigg|\,\nabla_{\!x}\,\frac{\,\big(\,e^{-i\fh_N t}\psi_\eta\big)(x)\,}{1-\omega_N(x)}  \,\bigg|^2 \!+ \; C\Big(\frac{\:(\log N^2t)^6}{\,N^2 t}\ell^3+t^2+\ell^3 N t + \ell^5 \Big)
\end{equation}
for all times $t>0$. Here we defined
\begin{equation}\label{eq:fheta}
\fh_N = -2 \Delta_x + V_N (x), \qquad \text{and} \quad \psi_{\eta} (x) = \ph (\eta+x/2) \ph(\eta-x/2)
\end{equation}
($\eta$ and $x$ denote, in other words, the center of mass and, respectively, the relative coordinate of particle one and two).

\medskip

Last, we consider the two-body integral contained in our bound (\ref{bound_on_KN}) to $\cK_N$. It has the natural interpretation of a quantity which tracks the dynamical formation of a short-scale structure in the evolution of the two-body initial factor state $\varphi^{\otimes 2}$. We study this problem in Section \ref{sec:dinamicalformation-2body}. The result (Corollary \ref{corollary1bodyanalysis}) is
\begin{equation}\label{bound_on_2body_piece}
\int_{\mathbb{R}^3\times\mathbb{R}^3}\!\rd \eta\,\rd x\,\theta_{\ell}(x)\bigg|\nabla_{\!x}\,\frac{\,\big(\,e^{-i\fh_N t}\psi_\eta\big)(x)\,}{1-\omega_N(x)}  \,\bigg|^2 \; \leqslant \; C\,\bigg(\frac{\,(\log N^2t)^6}{\;N^2 t}N^2\ell^3 + N t^2 + \ell^3   \bigg)
\end{equation}
for all times $t>0$.

\medskip

Bounds (\ref{bound_on_GN}), (\ref{bound_on_KN}), and (\ref{bound_on_2body_piece}) complete estimate (\ref{firstest_Ft}) for $\cF_N$. It takes the form
\begin{equation}\label{secondest_Ft}
\cF_N(t) \; \leqslant \; C\bigg( (N\log N)^{\frac{4}{5}}\,t^{2} + \frac{\:(\log N^2t)^6}{\,N^2 t}\ell^3+ t^2+\ell^3 N t + \ell^5 + \frac{\,(\log N^2t)^6}{\;N^2 t}N^2\ell^5 + \ell^2 N t^2 \bigg)
\end{equation}
for $t>0$ and $N^{-1}<\ell\ll N^{-\frac{2}{5}}$.
Using that, by Lemma \ref{lemma:asymptotics_t=0}, $\cF_N(0)\sim\ell/N^2$ and restricting to $\ell\ll N^{-1/2}$, $t\ll N^{-1}$, one has
\begin{equation}
\cF_N (t) \; \leqslant \; C\,\cF_N(0)\bigg(\:\frac{(\log N)^{\frac{4}{5}}}{N^{\frac{1}{5}}}\frac{(N^2t)^2}{N\ell}+\frac{\:(N\ell)^4}{N^2t}(\log N^2t)^6 \bigg)\, .
\end{equation}
\end{proof}

We compute now the asymptotics of $\cF_N$ at time $t=0$.
\begin{lemma}\label{lemma:asymptotics_t=0}
Let $\cF_N$ be the quantity defined in (\ref{FN}). Assume that $\ell$ scales with $N$ in such a way that $N\ell\to c_0$ when $N\to\infty$, where $c_0\in(0,+\infty]$. Then there exists a constant $C_\chi$, depending only on the cut-off function $\chi$ defined in (\ref{def_cutoff}) and on the potential $V$, such that
\begin{equation}\label{eq:FN0_asymptotics}
\lim_{\substack{N\to\infty \\ \!\!N\ell\to c_0 }}\frac{\;N^2}{\ell}\cF_N(0) \; = \; C_\chi \|\varphi\|_4^4\,.
\end{equation}
Moreover, if $N\ell\to\infty$, 
\begin{equation}\label{eq:FN0_asymptotics_l-large}
\lim_{\substack{N\to\infty \\ \!\!N\ell\to\infty }}\frac{\;N^2}{\ell}\cF_N(0) \; = \;4\pi a^2\,\|\chi\|_{L^1(\mathbb{R})}\|\varphi\|_{L^4(\mathbb{R}^3)}^4\,.
\end{equation}
\end{lemma}

\medskip

\begin{proof}
One has
\begin{equation}
\begin{split}
\frac{\,N^2}{\ell}\cF_N(0)\; & =\;\frac{\,N^2}{\ell}\int_{\mathbb{R}^6}\theta_\ell(x_1-x_2)\,\bigg|\,\frac{\,\varphi(x_1)\varphi(x_2)\,}{1-\omega_N(x_1-x_2)}-\varphi(x_1)\varphi(x_2)\,\bigg|^2\rd x_1\rd x_2 \\
&=\;\frac{\,N^2}{\ell}\inthree\rd x_2 \,|\varphi(x_2)|^2\inthree\rd x\,\theta_\ell(x)\bigg|\,\frac{\,\omega_N(x)\,}{1-\omega_N(x)}\,\bigg|^2\,|\varphi(x_2+x)|^2
\\
&=\;\frac{1}{N\ell}\inthree\rd x_2 \,|\varphi(x_2)|^2\inthree\rd x\,\chi\Big(\frac{\,|x|\,}{N\ell}\Big)\bigg|\,\frac{\,\omega(x)\,}{1-\omega(x)}\,\bigg|^2\,\Big|\varphi\Big(x_2+\frac{x}{N}\Big)\Big|^2
\\
& \xrightarrow[]{N\to\infty}\,C_\chi \|\varphi\|_{L^4(\mathbb{R}^3)}^4
\end{split}
\end{equation}
by continuity of $\varphi$ and by dominated convergence, where
\begin{equation}
C_\chi \; := \lim_{N\ell\to c_0 } \;  \frac{1}{\,N\ell}\inthree\rd x\,\chi\Big(\frac{\,|x|\,}{\,N\ell}\Big)\bigg|\,\frac{\,\omega(x)\,}{1-\omega(x)}\,\bigg|^2\,.
\end{equation}
The above limit clearly exists if $c_0$ is finite. If, instead, $N\ell\to\infty$ one has
\begin{equation}\label{thetaomega2}
\frac{1}{N\ell}\inthree\rd x\,\chi\Big(\frac{\,|x|\,}{N\ell}\Big)\Big|\,\frac{\,\omega(x)\,}{1-\omega(x)}\,\Big|^2\xrightarrow[]{N\ell\to\infty}\,4\pi a^2\,\|\chi\|_{L^1(\mathbb{R})}\,.
\end{equation}
To prove (\ref{thetaomega2}), one sees that integration for $|x|>1$ gives the leading contribution, since
\begin{equation}
\frac{1}{N\ell}\int_{|x|<1}\rd x\,\chi\Big(\frac{\,|x|\,}{N\ell}\Big)\Big|\,\frac{\,\omega(x)\,}{1-\omega(x)}\,\Big|^2\;\leqslant\;\frac{1}{N\ell}\frac{4\pi}{3}\frac{\,\|\omega\|_{\infty}^2}{\,(1-c\rho)^2}\;\xrightarrow[]{N\ell\to\infty}\,0
\end{equation}
while, when $|x|>1$, since $\omega(x)=a/|x|$,
\begin{equation}
\begin{split}
\frac{1}{N\ell}\int_{|x|>1}\rd x\,\chi\Big(\frac{\,|x|\,}{N\ell}\Big)\Big|\,\frac{\,\omega(x)\,}{1-\omega(x)}\,\Big|^2 & =\;\frac{\,a^2}{N\ell}\int_{|x|>1}\rd x\,\frac{\,\chi(\frac{\,|x|\,}{N\ell})}{(|x|-a)^2} \\
& = \; 4\pi a^2\int_{(N\ell)^{-1}}^2\!\rd r\, \chi(r)\,\frac{r^2}{(r-\frac{a}{N\ell})^2} \\
& \xrightarrow[]{N\ell\to\infty}\,4\pi a^2\,\|\chi\|_{L^1(\mathbb{R})}
\end{split}
\end{equation}
and (\ref{eq:FN0_asymptotics_l-large}) is proved.
\end{proof}

\section{Many-body problem for short times}\label{sec:Nbodyproblem}

In this section we control the growth of the nonnegative quantity $\cG_N$ defined in (\ref{gN}).

\medskip

\begin{proposition}\label{prop:GN}
Assume that $\ell\ll N^{-2/5}$. 
Then,
\begin{equation}\label{GNfinal}
\cG_N(t) \; \leqslant \; C\,(N\log N)^{\frac{4}{5}}\,t^{2}
\end{equation}
for all times $t\geqslant 0$.
\end{proposition}

\medskip

\begin{proof} For arbitrary $\tilde \ell \geq 2 \ell$ we have
\begin{equation}\begin{split}
\cG_N(t) & = \inthreeN\theta^2_{2\ell}(x_1-x_2)\big|\Psi_{N,t}(\bx)-\widetilde\Psi_{N,t}(\bx)\big|^2\rd\bx  \\
& \leqslant \; \inthreeN\theta^2_{\tilde\ell}(x_1-x_2)\big|\Psi_{N,t}(\bx)-\widetilde\Psi_{N,t}(\bx)\big|^2\rd\bx \; =: \; \widetilde\cG_N(t)\,.
\end{split}
\end{equation}
The parameter $\tilde \ell$ will be fixed later on (we will choose $\tilde\ell=(N \log N)^{-2/5}$).
Let us denote by $\langle\cdot,\cdot\rangle$ the scalar product in $L^2(\mathbb{R}^{3N})$ and by $\theta_{12}$, $\omega_{12}$, and $V_{ij}$ the operators of multiplication by $\theta_{\tilde\ell}(x_1-x_2)$, $\omega_N(x_1-x_2)$, and $V_N(x_i-x_j)$ respectively.
Then
\begin{equation}\label{G'}
\begin{split}
\frac{\rd\widetilde\cG_N(t)}{\rd t}\; & = \; \frac{\rd}{\rd t}\big\langle\Psi_{N,t}-\widetilde\Psi_{N,t}\,,\,\theta^2_{12}\,(\Psi_{N,t}-\widetilde\Psi_{N,t})\big\rangle \\
& = \; \big\langle \Psi_{N,t}-\widetilde\Psi_{N,t}\,,\,[iH_N,\theta^2_{12}]\,(\Psi_{N,t}-\widetilde\Psi_{N,t})\big\rangle \\
& \qquad + 2\,\text{Im} \,\big\langle\,\theta^2_{12}(\Psi_{N,t}-\widetilde\Psi_{N,t})\,,(H_N-\widetilde H_N)\widetilde\Psi_{N,t}\big\rangle \\
& \equiv \;\cJ_N^{(1)}(t)+\cJ_N^{(2)}(t)\,.
\end{split}
\end{equation}

\medskip

Since
\begin{equation}
[H_N,\theta^2_{12}] \; = \; [-\Delta_1-\Delta_2,\theta_{12}^2] \;  = \; -2\!\sum_{r=1,2} \Big(\nabla_r\cdot(\nabla_r\theta_{12})\,\theta_{12}+\theta_{12} (\nabla_r\theta_{12}) \cdot\nabla_r \Big)\,,
\end{equation}
the summand $\cJ_N^{(1)}$ on the r.h.s.~of (\ref{G'}) takes the form
\begin{equation}
\begin{split}
\big\langle \Psi_{N,t}-\widetilde\Psi_{N,t}\,,\,&[iH_N,\theta^2_{12}]\,(\Psi_{N,t}-\widetilde\Psi_{N,t})\big\rangle\;= \\
& =2\,i\!\sum_{r=1,2}\Big(\big\langle\nabla_r\,\delta\Psi_{N,t}\,,\,(\nabla_r\theta_{12})\,\theta_{12}\,\delta\Psi_{N,t}\big\rangle - \langle\,\delta\Psi_{N,t}\,,\,\theta_{12}(\nabla_r\theta_{12})\nabla_r\,\delta\Psi_{N,t}\big\rangle\Big)
\end{split}
\end{equation}
having set $\delta\Psi_{N,t}:=\Psi_{N,t}-\widetilde\Psi_{N,t}$.
Then
\begin{equation}\label{1_summ_G'}
\begin{split}
\big|\,\cJ_N^{(1)}(t)\,\big| & \leqslant \;4\sum_{r=1,2}\,\Big|\,\big\langle(\nabla_r\theta_{12})\cdot(\nabla_r\delta\Psi_{N,t})\,,\,\theta_{12}\,\delta\Psi_{N,t}  \big\rangle \, \Big|\\
& \leqslant \; 8\,\sqrt{\widetilde\cG_N(t)\,}\,\bigg(\inthreeN\!\rd\bx\,\big|\nabla\theta_{12}\big|^2 \big|(\nabla_1 \delta\Psi_{N,t})(\bx)\big|^2 \bigg)^{\!\frac{1}{2}} \\
& = \; 8\,\sqrt{\widetilde\cG_N(t)\,}\;\big\|\nabla\theta_{12}\big\|_\infty\bigg(\inthreeN\!\rd\bx\,\big|(\nabla_1 \delta\Psi_{N,t})(\bx)\big|^2 \bigg)^{\!\frac{1}{2}} \\
& \leqslant \; C\,\sqrt{\widetilde\cG_N(t)\,}\,{\tilde\ell}^{-1}\,.
\end{split}
\end{equation}
On the last line above we used
\begin{equation}
\inthreeN\!\rd\bx\,\big|(\nabla_1 \delta\Psi_{N,t})(\bx)\big|^2 \; \leq C \,.
\end{equation}
Indeed, by the symmetry of $\Psi_{N,t}$, and because of (\ref{expectation-H}),
\begin{equation}
\begin{split}
\inthreeN\!\rd\bx\,\big|\nabla_1 \Psi_{N,t}\big|^2 & = \; \big\langle\Psi_{N,t},(-\Delta_1)\Psi_{N,t}\big\rangle \; \leqslant \; \big\langle\varphi^{\otimes N},\frac{H_N}{\,N}\,\varphi^{\otimes N}\big\rangle \; \leq C ,
\end{split}
\end{equation}
while, due to the factorization (\ref{Psi_Ntilde}),
\begin{equation}
\begin{split}
\inthreeN\!\rd\bx\,\big|\nabla_1 \widetilde\Psi_{N,t}\big|^2 & = \int_{\mathbb{R}^6}\rd x_1\rd x_2 \,\big|(\nabla_1 e^{-it \fh_N^{(1,2)}}\!\varphi^{\otimes 2})(x_1,x_2)\big|^2 \; \leqslant \; \big\langle\varphi^{\otimes 2},\fh_N^{(1,2)}\!\varphi^{\otimes 2}\big\rangle \; \leq\; C \, .
\end{split}
\end{equation}

\medskip

Let us now examine the summand $\cJ_N^{(2)}$ in the r.h.s.~of (\ref{G'}). One has
\begin{equation}
\begin{split}
\big|\,\cJ_N^{(2)}(t)\,\big|\; & \leqslant \; 2\,\big|\,\big\langle\,\theta^2_{12}(\Psi_{N,t}-\widetilde\Psi_{N,t})\,,(H_N-\widetilde H_N)\widetilde\Psi_{N,t}\big\rangle\,\big| \\
& = \; 2\,\big|\,\langle\,\delta\Psi_{N,t}\,,\theta^2_{12}\,U_N^{(1,2)}\,\widetilde\Psi_{N,t}\,\rangle\,\big| \\
& \leqslant \; 4\,\big|\,\langle\,\delta\Psi_{N,t}\,,\theta^2_{12}\,\Big(\sum_{j=3}^N V_{1j} \Big)\,\widetilde\Psi_{N,t}\,\rangle\,\big|
\end{split}
\end{equation}
(last line following by the permutational symmetry $1\leftrightarrow 2$), whence, by Schwarz inequality,
\begin{equation}\label{J2intermediate}
\begin{split}
\big|\,\cJ_N^{(2)}(t)\,\big|\; & \leqslant \; 4\,\big\langle\,\delta\Psi_{N,t}\,,\theta^2_{12}\,\delta\Psi_{N,t}\,\big\rangle^{\frac{1}{2}}\bigg(\,\sum_{j,k=3}^N\big\langle\widetilde\Psi_{N,t}\,, \theta^2_{12} V_{1k}V_{1j}\,\widetilde\Psi_{N,t}\,\big\rangle\bigg)^{\!\frac{1}{2}} \\
& \leqslant \; C\,\sqrt{\widetilde\cG_N(t)\,}\Big(\,N\big\langle\widetilde\Psi_{N,t}\,, \theta^2_{12} V_{13}^2\,\widetilde\Psi_{N,t}\,\big\rangle
+ N^2
\big\langle\widetilde\Psi_{N,t}\,, \theta^2_{12} V_{13}V_{14}\,\widetilde\Psi_{N,t}\,\big\rangle\,\Big)^{\!\frac{1}{2}}\,.
\end{split}
\end{equation}

\medskip

The first term on the r.h.s. is estimated by
\begin{equation}
\begin{split}
\big\langle\widetilde\Psi_{N,t}\,, \theta^2_{12} V_{13}^2\,\widetilde\Psi_{N,t}\,\big\rangle
& =  \;\inthreeN\theta^2_{\tilde\ell}(x_1-x_2)\,V_N^2(x_1-x_3)\,|\psi_t(x_1,x_2)|^2\,|\Phi_t(x_3,\dots,x_N)|^2\,\rd\bx \\
& \leqslant \; \| \psi_t \|^2_{\infty} \, \!\! \inthreeN\theta^2_{\tilde\ell}(x_1-x_2)\, N^4V^2(N(x_1-x_3))\,|\Phi_t(x_3,\dots,x_N)|^2\,\rd\bx \\
& = \; C \| \psi_t \|^2_{\infty} \, N^4 \tilde\ell^{\,3} \!\! \int_{\mathbb{R}^{3(N-1)}}\!V^2(Nx)\,|\Phi_t(x_3,\dots,x_N)|^2\,\rd x \rd x_3 \cdots \rd x_N \\
& = \; C\,\| \psi_t \|^2_{\infty}  \, N \tilde\ell^{\,3} \|\Phi_t\|_2^2 \\
& = \; C\,\| \psi_t \|^2_{\infty}  \, N \tilde\ell^{\,3} \,.
\end{split}
\end{equation}
It follows by Proposition \ref{prop:t-x-boundedness} that, under the assumption (\ref{assumption-on-phi}),
\begin{equation}\label{psi_t-boundedness}
\| \psi_t \|^2_{L^\infty(\mathbb{R}^6,\rd x_1\rd x_2)} \leq C (\log N) \,
\end{equation}
and thus that
\begin{equation}
\big\langle\widetilde\Psi_{N,t}\,, \theta^2_{12} V_{13}^2\,\widetilde\Psi_{N,t}\,\big\rangle
\; \leqslant \; C\,N(\log N)^2\,\tilde\ell^{\,3}\,.
\end{equation}

\medskip

On the other hand, the second term on the r.h.s.~of (\ref{J2intermediate}) is estimated as
\begin{equation}
\begin{split}
\big\langle&\widetilde\Psi_{N,t} \,, \theta^2_{12} V_{13}V_{14}\,\widetilde\Psi_{N,t}\,\big\rangle \; = \\
& = \;\inthreeN\theta^2_{\tilde\ell}(x_1-x_2)\,V_N(x_1-x_3)V_N(x_1-x_4)\,|\psi_t(x_1,x_2)|^2\,|\Phi_t(x_3,\dots,x_N)|^2\,\rd\bx \\
& \leqslant \; C\,\| \psi_t \|^2_{\infty}\,\tilde\ell^{\,3}\!\int_{\mathbb{R}^{3(N-1)}}V_N(x_1-x_3)V_N(x_1-x_4)\,|\Phi_t(x_3,\dots,x_N)|^2\,\rd x_1\rd x_3\cdots\rd x_N \\
& \leqslant \; C\,\| \psi_t \|^2_{\infty}\,\tilde\ell^{\,3}\|V_N\|_{\frac{3}{2}}\!\int_{\mathbb{R}^{3(N-2)}}\!\!\rd x_3\cdots\rd x_N\,\big|\sqrt{\mathbbm{1}-\Delta_3}\,\Phi_t(x_3,\dots,x_N)\big|^2\!\inthree\rd x_1\,V_N(x_1-x_4) \\
& \leqslant \; C\,\| \psi_t \|^2_{\infty}\,\|V\|_{\frac{3}{2}}\|V\|_1\frac{\;\tilde\ell^{\,3}}{N} \, \Big\langle\Phi_t,
\Big(\mathbbm{1}+\frac{\,H_N^{(3)}}{N-2}\Big)\,\Phi_t\Big\rangle_{L^2(\mathbb{R}^{3(N-2)})} \\
& \leqslant \; C\,\frac{\,(\log N)^2}{\;N}\tilde\ell^{\,3}\,.
\end{split}
\end{equation}
In the last inequality we used again bound (\ref{psi_t-boundedness}) and the asymptotics (\ref{expectation-H}).
Thus, (\ref{J2intermediate}) reads
\begin{equation}\label{2_summ_G'}
\big|\,\cJ_N^{(2)}(t)\,\big| \; \leqslant C\,\sqrt{\widetilde\cG_N(t)\,}\,N\tilde\ell^{\,\frac{3}{2}}\log N \,.
\end{equation}

\medskip

Altogether, (\ref{G'}), (\ref{1_summ_G'}), and (\ref{2_summ_G'}) give
\begin{equation}
\big|\,\widetilde\cG_N'(t)\,\big| \;\leqslant\; C\,\sqrt{\widetilde\cG_N(t)\,}\,\big(\tilde\ell^{\,-1}+N\tilde\ell^{\,\frac{3}{2}}\log N  \big)\,. 
\end{equation}
Letting $\tilde\ell=(N\log N)^{-2/5}$, we get
\begin{equation}
\big|\,\widetilde\cG_N'(t)\,\big| \;\leqslant\; C\,\sqrt{\widetilde\cG_N(t)\,}\,(N\log N)^{\frac{2}{5}}\,
\end{equation}
which implies (\ref{GNfinal}) by Gronwall Lemma, because $\widetilde\cG_N(0)=0$.
\end{proof}

\section{Reduction to the two-body problem}\label{sec:2bodyproblem}

The goal of this section is to reduce the study of the quantity $\cK_N$, defined in (\ref{kN}), to the analysis of a two-body term (which will then be controlled in Section \ref{sec:dinamicalformation-2body}). We will use, in this section, the coordinates $(\eta,x)$ defined by \begin{equation}\begin{split} \eta &= (x_1 + x_2)/2 \qquad \; \text{(center of mass coordinate)} \\ x &= x_2 -x_1 \qquad\qquad  \text{(relative coordinates).}
\end{split}\end{equation} In these coordinates, the two-body Hamiltonian $\fh_N^{(1,2)}$ introduced in (\ref{hN12}) takes the form \begin{equation} \fh_N^{(1,2)} = - \Delta_{\eta}/2 + \fh_N , \qquad \text{with} \qquad  \fh_N = -2 \Delta_x + V_N (x).\end{equation} Note, also, that the two-body initial data $\ph^{\otimes 2}$ is given by
\begin{equation}\label{eq:psietax}
\psi(\eta,x) \; = \; \psi_\eta (x) \; = \; \varphi\Big(\eta+\frac{x}{2}\Big)\,\varphi\Big(\eta-\frac{x}{2}\Big) \,.
\end{equation}

\begin{proposition}\label{prop:reduction}
Suppose that the assumptions of Theorem \ref{N-thm} are satisfied. Let $\fh_N = -2\Delta + V_N (x)$ and $\psi_{\eta} (x)$ be defined as in (\ref{eq:psietax}). Then, if $\cK_N (t)$ is defined as in (\ref{firstest_Ft}), we have
\begin{equation}\label{bound_on_KN_plus_remainders}
\begin{split}
\cK_N(t) \; & \leqslant \; C\ell^2\!\int_{\mathbb{R}^3\times\mathbb{R}^3}\!\rd \eta\,\rd x\,\theta_{2\ell} \,(x)\bigg|\,\nabla_{\!x}\,\frac{\,\big(\,e^{-i\fh_N t}\psi_\eta\big)(x)\,}{1-\omega_N(x)}  \,\bigg|^2 \!+ \; C\Big(\frac{\:(\log N^2t)^6}{\,N^2 t}\ell^3+t^2+\ell^3 N t + \ell^5 \Big)
\end{split}
\end{equation}
for all times $t>0$.
\end{proposition}

\begin{proof}
We divide the proof in four steps.

\medskip

{\it Step 1.} We have
\begin{equation}\label{1remainder}
\begin{split}
\cK_N(t) \;  \leqslant \; & 2\int_{\mathbb{R}^3\times\mathbb{R}^3}\theta_\ell(x_1-x_2)\left|\,
\frac{\psi_t(x_1,x_2)}{1-\omega_N (x_1-x_2)}-\varphi(x_1)\varphi(x_2)\, \right|^2
\rd x_1\rd x_2 \; + \; C \ell^3 N t
\end{split}
\end{equation}
where $\psi_t(x_1,x_2)=\exp(-it\fh_N^{(1,2)})\varphi^{\otimes 2}$ is defined in (\ref{Psi_Ntilde}).

\medskip

To prove (\ref{1remainder}), we use that $\widetilde\Psi_{N,t}=\psi_t\otimes\Phi_t$ and we split
\begin{equation}
\begin{split}
\cK_N(t) \; &  = \; \inthreeN\theta_\ell(x_1-x_2)\,\bigg|\,\frac{\,\widetilde\Psi_{N,t}(\bx)\,}{1-\omega_N(x_1-x_2)}-\prod_{i=1}^N\varphi(x_i)\,\bigg|^2\rd \bx \\
& \leqslant\; 2\int_{\mathbb{R}^3\times\mathbb{R}^3}\theta_\ell(x_1-x_2)\left|\,\frac{\psi_t(x_1,x_2)}{1-\omega_N(t)(x_1-x_2)}-\varphi(x_1)\varphi(x_2)\, \right|^2\rd x_1\rd x_2 \\
& \qquad + 2\inthreeN\theta_\ell(x_1-x_2)\Big|\,\big(\varphi^{\otimes 2}\otimes\Phi_t\big)(\bx) - \varphi^{\otimes N}(\bx)\,\Big|^2\rd \bx
\end{split}
\end{equation}
where
\begin{equation}
\begin{split}
\inthreeN\theta_\ell(x_1-x_2)& \Big|\,\big(\varphi^{\otimes 2}\otimes\Phi_t\big)(\bx) - \varphi^{\otimes N}(\bx)\,\Big|^2\rd \bx\;  = \\
& = \;  \int_{\mathbb{R}^3\times\mathbb{R}^3}\theta_\ell(x_1-x_2)|\varphi (x_1)|^2|\varphi (x_2)|^2\,\Big\|\,\Phi_t-\varphi^{\otimes(N-2)}\Big\|^2_{L^2(\mathbb{R}^{3(N-2)})} \rd x_1 \rd x_2 \\
& \leqslant \; C\,\ell^3\,\|\varphi\|_{L^4(\mathbb{R}^3)}^4\,\big\|\Phi_t-\Phi_{t=0}\big\|^2_{L^2(\mathbb{R}^{3(N-2)})}\,.
\end{split}
\end{equation}
Eq. (\ref{1remainder}) follows now from
\begin{equation}\label{eq:Phi} \big\|\Phi_t-\Phi_0\big\|_2^2 \; \leq \; C N t \, . \end{equation}
To show (\ref{eq:Phi}), observe that
\begin{equation}
\begin{split}
\Big|\,\frac{\rd}{\rd t} \big\|\Phi_t-\Phi_{0}\big\|_2^2\,\Big|\; & \leqslant\; 2\,\big|\,\big\langle\Phi_0,H_N^{(3)}\Phi_t\big\rangle\,\big| \leq \; 2 \, \big\langle\Phi_0,H_N^{(3)}\Phi_0\big\rangle \; \leqslant \; CN\, .
\end{split}
\end{equation}
The last inequality follows from (\ref{expectation-H}) with $N$ replaced by $N-2$ (recall the definition of $H_N^{(3)}$ from (\ref{HN3})).

\bigskip

{\it Step 2.}
We have
\begin{equation}
\label{eq:cK-2st_remainder}
\begin{split}
\int_{\mathbb{R}^3\times\mathbb{R}^3} \theta_\ell(x_1-x_2) &\left|\,
\frac{\psi_t(x_1,x_2)}{1-\omega_N (x_1-x_2)}-\varphi(x_1)\varphi(x_2)\, \right|^2
\rd x_1\rd x_2 \\ \; \leq \; & C\ell^2\!\int_{\mathbb{R}^3\times\mathbb{R}^3}\!\rd \eta\,\rd x\,\theta_{2\ell}(x)\bigg|\,\nabla_{\!x}\,\frac{\,\big(e^{-i\fh_N t}\psi_\eta\big)(x)\,}{1-\omega_N(x)}  \,\bigg|^2 \;+\;C \left( \cR^{(1)}_N(t) + \cR^{(2)}_N (t)\right)
\end{split}
\end{equation}
with
\begin{equation}\label{eq:R1Nt}
\cR^{(1)}_N(t) \; := \inthree\rd \eta \! \inthree\rd x \, \theta_\ell(x)\,\bigg| \inthree\rd y \,\widetilde\theta_{\ell}(y) \frac{\,\big(\,e^{-i\fh_N t} \omega_N \,\psi_{\eta}\big)(y)\,}{1-\omega_N (y)}\,\bigg|^2
\end{equation}
and
\begin{equation}\label{eq:R2Nt}
\cR^{(2)}_{N} (t) := \!\inthree\rd \eta \! \inthree\rd x \, \theta_\ell(x)\,\bigg|\,e^{-i\frac{\,\Delta_\eta}{2} t}\psi(\eta,x)- \inthree\rd y \,\widetilde\theta_{\ell}(y)\frac{\,\big(\,e^{-i\fh_N t}(1-\omega_N)\psi_{\eta}\big)(y)\,}{1-\omega_N (y)}\,\bigg|^2\,.
\end{equation}
Here we use the notation
\begin{equation}
\widetilde\theta_\ell(y) \; := \; \frac{\theta_\ell(y)}{\;\|\theta_\ell\|_1}\,.
\end{equation}

\medskip

To prove (\ref{eq:cK-2st_remainder}), we observe that
\begin{equation}\label{split_to_poinc}
\begin{split}
&\int_{\mathbb{R}^3\times\mathbb{R}^3} \rd x_1\rd x_2 \, \theta_\ell(x_1-x_2)\left|\,\frac{\psi_t(x_1,x_2)}{1-\omega_N(x_1-x_2)}-\varphi(x_1)\varphi(x_2)\, \right|^2  \\
& = \; \int_{\mathbb{R}^3\times\mathbb{R}^3}\!\rd \eta\,\rd x\,\theta_\ell(x)\,\bigg|\,\frac{\,e^{\,i\frac{\;\Delta_\eta}{2} t}\,(e^{-i\fh_N t }\psi_\eta)(x)}{1-\omega_N(x)}-\psi(\eta,x) \,\bigg|^2 \\
& \leqslant \; 2\int_{\mathbb{R}^3\times\mathbb{R}^3}\!\rd \eta\,\rd x\,\theta_\ell(x)\,\bigg|\, \frac{\,e^{\,i\frac{\,\Delta_\eta}{2} t}\,(e^{-i\fh_N t }\psi_\eta)(x)}{1-\omega_N(x)}-  \inthree\rd y \,\widetilde\theta_\ell(y)\frac{\,e^{i\frac{\;\Delta_\eta}{2} t}\,(e^{-i\fh_N t }\psi_\eta)(y)}{1-\omega_N(y)}          \,\bigg|^2 \\
& \qquad + 2\int_{\mathbb{R}^3\times\mathbb{R}^3}\!\rd \eta\,\rd x\,\theta_\ell(x)\,\bigg|\,\psi(\eta,x)-  \inthree\rd y \,\widetilde\theta_\ell(y)\frac{\,e^{\,i\frac{\,\Delta_\eta}{2} t}\,(e^{-i\fh_N t }\psi_\eta)(y)}{1-\omega_N(y)}  \,\bigg|^2\!\! \\
& \leqslant \; 2\int_{\mathbb{R}^3\times\mathbb{R}^3}\!\rd \eta\,\rd x\,\theta_\ell(x)\,\bigg|\, \frac{\,(e^{-i\fh_N t}\psi_\eta)(x)}{1-\omega_N(x)}-  \inthree\rd y \,\widetilde\theta_\ell(y)
\frac{\,(e^{-i\fh_N t }\psi_\eta)(y)}{1-\omega_N(y)}\,\bigg|^2 \\
& \qquad + 2\int_{\mathbb{R}^3\times\mathbb{R}^3}\!\rd \eta\,\rd x\,\theta_\ell(x)\,\bigg|\,e^{\,-i\frac{\,\Delta_\eta}{2} t} \psi(\eta,x)-  \inthree\rd y \,\widetilde\theta_\ell(y)\frac{\,(e^{-i\fh_N t }\psi_\eta)(y)}{1-\omega_N(y)}  \,\bigg|^2\!\!
\end{split}
\end{equation}
by unitarity of $e^{i\Delta_{\eta} t/2}$. The second term can be clearly bounded by the sum of $\cR^{(1)}_{N} (t)$ and $\cR^{(2)}_N (t)$.
For every fixed $\eta \in \bR^3$, the first term on the r.h.s.
of the last equation can be estimated using the Poincar{\'e} inequality as
\begin{equation}
\begin{split}
\inthree\rd x\,\theta_\ell(x)&\,\bigg|\, \frac{\,(e^{-i\fh_N t }\psi_\eta)(x)}{1-\omega_N(x)}-  \inthree\rd y \,\widetilde\theta_\ell(y)
\frac{\,(e^{-i\fh_N t }\psi_\eta)(y)}{1-\omega_N(y)}\,\bigg|^2  \\
& \leqslant \; \int_{\supp\theta_\ell}\rd x\,\bigg|\, \frac{\,(e^{-i\fh_N t }\psi_\eta)(x)}{1-\omega_N(x)}-  \int_{\supp\theta_\ell}  \rd y \,\widetilde\theta_\ell(y)\frac{\,(e^{-i\fh_N t }\psi_\eta)(y)}{1-\omega_N(y)}          \,\bigg|^2 \\
& \leqslant \; C\ell^2\!\int_{\supp\theta_\ell}\rd x\,\bigg|\nabla_{\!x} \frac{\,(e^{-i\fh_N t }\psi_\eta)(x)}{1-\omega_N(x)} \,\bigg|^2 \\
& \leqslant \;C\ell^2\!\inthree\rd x\,\theta_{2\ell}(x)\bigg|\nabla_{\!x} \frac{\,(e^{-i\fh_N t }\psi_\eta)(x)}{1-\omega_N(x)} \,\bigg|^2\,.
\end{split}
\end{equation}

\bigskip

{\it Step 3.} Suppose that $\cR_N^{(1)} (t)$ is defined as in (\ref{eq:R1Nt}). Then
\begin{equation}\label{eq:step3}
\cR_N^{(1)}(t)  \; \leqslant \;  C\ell^3\frac{\:(\log N^2t)^6}{\,N^2 t}\,.
\end{equation}

\medskip

To show (\ref{eq:step3}), we note that, from (\ref{eq:R1Nt}),
\begin{equation}\label{RN1_estimate}
\cR_N^{(1)}(t)  \;  \leqslant \;  C\ell^3\!\inthree\rd \eta \,\big\|\,e^{-i\fh_N t}\omega\,\psi_{\eta}\, \big\|_{L^{\infty}(\mathbb{R}^3,\rd x)}^2\,.
\end{equation}
Let $\Omega_N$ be the wave operator associated with the Hamiltonian $\fh_N$, defined as the strong limit \[ \Omega_N = s-\lim_{t\to \infty} e^{i\fh_N t} e^{2i\Delta t} \,. \]
Then, by the intertwining property (\ref{eq:intertw}) and Yajima's bound (\ref{yb}),
\begin{equation}\label{jb_on_psiXN}
\big\|\,e^{-i\fh_N t}\omega_N \,\psi_{\eta}\, \big\|_\infty^2  \; \leqslant \; C\, \big\|\,e^{\,2i t \Delta }\Omega_N^*\omega_N\,\psi_{\eta}\, \big\|_\infty^2 \,.
\end{equation}
In Proposition \ref{prop-Omegaomegapsilambda-dispersive} we prove that
\begin{equation}\label{de_for_psiXN}
\big\|\,e^{2it \Delta }\Omega_N^*\omega_N\,\psi_{\eta}\, \big\|_\infty^2  \;\leqslant\;c_s\frac{\,\tri\psi_\eta\tri^2}{\,(N^2 t)^{3/s}}\qquad\qquad\forall\,s\in(3,+\infty]
\end{equation}
where $c_s\sim (s-3)^{-6}$ as $s\to 3^+$ and where we defined the norm \begin{equation}
\label{eq:trinorm} \tri\psi_\eta\tri=\|\psi_\eta\|_{W^{3,1}}+\|\psi_\eta\|_{W^{3,\infty}}.
\end{equation}
Optimizing in $s>3$, we have
\begin{equation}
\cR_N^{(1)}(t)  \; \leqslant \;  C\ell^3\frac{\:(\log N^2t)^6}{\,N^2 t}\inthree\rd \eta \,\tri\psi_{\eta}\tri^2\,.
\end{equation}
and thus, since $\tri \psi_{\eta} \tri \leq C \langle \eta \rangle^{-\alpha}$ for some $\alpha >3$ (by the definition (\ref{eq:psietax}), and the condition (\ref{assumption-on-phi})), we obtain (\ref{eq:step3}).

\bigskip

{\it Step 4.} Assume that $\cR^{(2)}_N (t)$ is defined as in (\ref{eq:R2Nt}). Then \begin{equation}\label{eq:step4}
\cR_N^{(2)}(t) \; \leqslant \; C\,\big(t^2+\ell^3 t + \ell^5\big)\,.
\end{equation}

\medskip

First we rewrite
\begin{equation}
\begin{split}
\inthree\rd y  \,\widetilde\theta_{\ell}(y)\frac{\,\big(\,e^{-it\fh_N }(1-\omega_N)\psi_{\eta}\big)(y)\,}{1-\omega_N(y)} \; & = \; \inthree\rd y \,\widetilde\theta_{\ell}(y)\big(e^{-2i t\cL_N }\psi_{\eta}\big)(y)
\end{split}
\end{equation}
by means of the operator
\begin{equation}\label{defcL}
\cL_N \; := \;
-\Delta +2\,\frac{\,\nabla \omega_N}{1-\omega_N} \cdot\nabla\, .
\end{equation}
In fact, since $\fh_N (1-\omega_N)=(-2\Delta+V_N)(1-\omega_N)=0$, one has
\begin{equation}
\frac{1}{1-\omega_N}\, \fh_N \, (1-\omega_N) \phi\;=\;2\cL_N \,\phi \qquad\qquad\forall\,\phi\in H^2(\mathbb{R}^3)
\end{equation}
whence
\begin{equation}\label{eq:evL}
e^{-i \fh_N t}\, (1-\omega_N)\phi\;=\;(1-\omega_N)e^{-2 i\cL_N t}\phi\qquad\;\forall\,\phi\in L^2(\mathbb{R}^3)\,.
\end{equation}
It is worth noticing that
\begin{equation}\label{Lsimmetry}
\langle \phi,\cL_N \psi\rangle_\bullet = \langle\cL_N \phi,\psi\rangle_\bullet =  \langle\nabla\phi,\nabla\psi\rangle_\bullet
\end{equation}
where \[ \langle f , g \rangle_{\bullet} = \int \rd x \, (1-\omega_N (x))^2 \,
\overline{f} (x) \, g (x) \, . \] It follows that the operator $\cL_N$ is self-adjoint
on the weighted Hilbert space $L^2 (\bR^3, (1-\omega_N (x))^2 \rd x)$ ($\cL_N$ is the Laplacian on the weighted space). It is also important to note that, because of the properties of $\omega_N$, the norm $\| \cdot \|_{\bullet}$ defined by the weighted product $\langle \cdot, \cdot \rangle_{\bullet}$ is comparable with the standard $L^2$-norm, in the sense that $c \| \phi \|_2 \leq \|\phi\|_\bullet \leqslant\|\phi\|_2$, with an appropriate constant $c >0$.

\medskip

{F}rom (\ref{eq:R2Nt}), we find
\begin{equation}\label{R2split_terms}
\begin{split}
\cR_N^{(2)}(t) \; & \leqslant \; C\!\! \inthree\rd x \, \theta_\ell \, (x)\inthree\rd \eta \,\big|\,e^{-i\frac{\,\Delta_\eta}{2} t}\psi(\eta,x)-\psi(\eta,x)\big|^2 \\
& \qquad + C\!\inthree\rd \eta \! \inthree\rd x \, \theta_\ell \,(x)\,\big|\psi_\eta(x)-\psi_\eta(0)\big|^2 \\
& \qquad + C\!\inthree\rd \eta \! \inthree\rd x \, \theta_\ell \,(x)\,\Big|\psi_\eta(0)-\inthree\rd y \,
\widetilde\theta_{\ell}(y)\psi_{\eta}(y)  \Big|^2 \\
& \qquad + C\!\inthree\rd \eta \! \inthree\rd x \, \theta_\ell \,(x)\,\Big|\inthree\rd y \,\widetilde\theta_{\ell}(y)\big(e^{-2i t\cL_N}-\mathbbm{1}\big)\psi_{\eta}(y)\,
\Big|^2 \,.
\end{split}
\end{equation}

\medskip

To estimate the first summand on the r.h.s.~of (\ref{R2split_terms}) we use that
\begin{equation}
\Big|\,\frac{\rd}{\rd t}  \big\|\,e^{-i\frac{\,\Delta_\eta}{2} t}\psi(\eta,x) -\psi(\eta,x)\,   \big\|^2_{L^2(\mathbb{R}^3,\rd \eta)}  \,\Big| \; \leqslant \; 2\,\|\nabla_{\!\eta}\psi(\eta,x)\|_{L^2(\mathbb{R}^3,\rd\eta)}^2 \; \leqslant \; C
\end{equation}
uniformly in $x\in \bR^3$ (the last inequality follows from (\ref{superest})), whence
\begin{equation}
\inthree\rd x \, \theta_\ell(x) \inthree\rd \eta \,\big|\,e^{-i\frac{\,\Delta_\eta}{2} t}\psi(\eta,x)-\psi(\eta,x)\big|^2 \; \leqslant \; C\ell^3 t\,.
\end{equation}

\medskip

To control the second summand on the r.h.s.~of (\ref{R2split_terms}), we observe that,
\begin{equation}\label{psiX-psi0}
\begin{split}
\big|\,\psi_\eta(x)-\psi_\eta(0)\,\big| \; & = \; \bigg|\int_0^1\rd s\frac{\,\rd}{\,\rd s} \psi_\eta(sx)\bigg|
\; \leqslant \; |x|\, \|\nabla_{\!x}\psi_\eta\|_{\infty} \; \leqslant \; C\frac{|x|}{\:\langle \eta\rangle^{\alpha}}
\end{split}
\end{equation}
(the last inequality follows from (\ref{superest})), and thus
\begin{equation}
\begin{split}
\inthree\rd \eta \! \inthree\rd x \, \theta_\ell \, (x)\,\big|\,\psi_\eta(x)-\psi_\eta(0)\,\big|^2\; & \leqslant\; C\inthree\frac{\rd \eta}{\;\,\langle \eta\rangle^{2\alpha}}\! \inthree\rd x \, \theta_\ell(x)\,|x|^2  \;\leqslant\; C\ell^5\,.
\end{split}
\end{equation}

\medskip

The third summand on the r.h.s.~of (\ref{R2split_terms}) can be bounded by
\begin{equation}
\begin{split}
\inthree\rd \eta \! \inthree\rd x \,& \theta_\ell(x)\,\Big|\,\psi_\eta(0)-\inthree\rd y \,\widetilde\theta_{\ell}(y)\,\psi_{\eta}(y)\Big|^2 \\
&  = \; \inthree\rd \eta \! \inthree\rd x \, \theta_\ell(x)\,\Big|\inthree\rd y \,\widetilde\theta_{\ell}(y)\big(\psi_{\eta}(y)-\psi_\eta(0)\big)\Big|^2 \\
& \leqslant\; C\inthree\!\frac{\rd \eta}{\;\,\langle \eta\rangle^{2\alpha}}\! \inthree\rd x \, \theta_\ell(x) \bigg(\inthree \rd y\, \widetilde\theta_{\ell}(y) |y| \bigg)^2 \; \leqslant \; C\ell^5
\end{split}
\end{equation}
where (\ref{psiX-psi0}) has been used.

\medskip

Finally, to estimate the fourth summand on the r.h.s.~of (\ref{R2split_terms}), we observe that, for fixed $\eta \in \bR^3$,
\begin{equation}\label{R2_4_a}
\bigg|\inthree\rd y\,\widetilde\theta_{\ell}(y)\big((e^{-2it\cL_N }-\mathbbm{1})\,\psi_{\eta}\big)(y)\,\bigg|^2 \; = \; \bigg|\,\Big\langle\,\widetilde\theta_\ell \,, \big(e^{-2it\cL_N}-\mathbbm{1}\big)\psi_{\eta}\Big\rangle
\,\bigg|^2
\leqslant \; \frac{C}{\ell^3}\,\big\| \big(e^{-2it \cL_N}-\mathbbm{1}\big)\psi_{\eta}  \big\|^2
\end{equation}
Expanding
\[ \big(e^{-2it \cL_N}-\mathbbm{1}\big)\psi_{\eta}  = -2i \int_0^t \rd s \, e^{-2is\cL_N} \cL_N \psi_{\eta} \]
we obtain
\begin{equation}
\begin{split}
\Big|\inthree\rd y\,\widetilde\theta_{\ell}(y)\big((e^{-2it\cL_N }-\mathbbm{1})\,\psi_{\eta}\big)(y)\,\Big|^2 \leq \; &\frac{Ct}{\ell^3} \int_0^t \rd s \| e^{-2is\cL_N} \cL_N \psi_{\eta} \|^2 \\ \leq \; &\frac{Ct}{\ell^3} \int_0^t \rd s \| e^{-2is\cL_N} \cL_N \psi_{\eta} \|_{\bullet}^2 \\ \leq \; &\frac{Ct^2}{\ell^3} \| \cL_N \psi_{\eta} \|_{\bullet}^2 \, .
\end{split}
\end{equation}
To bound the r.h.s. of the last equation we observe that
\begin{equation}\label{R2_4_b}
\begin{split}
\big\| \cL_N \,\psi_{\eta}  \big\|_\bullet \; & \leqslant \; \big\| \Delta \,\psi_{\eta}  \big\| + 2 \,\Big\| \nabla \omega_N \cdot\nabla \psi_{\eta} \Big\| \\
& \leqslant \; C \langle \eta\rangle^{-\alpha} + \int \rd x \, |\nabla \omega_N (x)|^2 \, |\nabla_x \psi_{\eta} (x)|^2 \\
& \leqslant \; C \langle \eta\rangle^{-\alpha}
\end{split}
\end{equation}
for $\alpha >3$. Here we used the definition (\ref{defcL}) of $\cL_N$ and Eq. (\ref{superest}). Plugging (\ref{R2_4_a}) and (\ref{R2_4_b}) into the fourth summand on the r.h.s.~of (\ref{R2split_terms}), we find
\begin{equation}
\inthree\rd \eta \! \inthree\rd x \, \theta_\ell(x)\,\Big|\inthree\rd y \,\widetilde\theta_{\ell}(y)\big((e^{-2it\cL_N}-\mathbbm{1})\,\psi_{\eta}\big)(y)\,\Big|^2 \; \leqslant \; C\,t^2.
\end{equation}

\medskip

The results of Steps 1-4 complete the proof of (\ref{bound_on_KN_plus_remainders}).
\end{proof}

\section{Dynamical formation of correlations among two particles}\label{sec:dinamicalformation-2body}

In this section we estimate the quantity
\begin{equation}\label{x1x2xX}
\int_{\mathbb{R}^3\times\mathbb{R}^3}\!\rd \eta\,\rd x\,\theta_{2\ell}(x)\bigg|\nabla_{\!x}\,\frac{\,\big(\,e^{-i\fh_N t}\psi_\eta\big)(x)\,}{1-\omega_N(x)}  \,\bigg|^2
\end{equation}
which arises in the bound (\ref{bound_on_KN_plus_remainders}).

\medskip

To control the integral (\ref{x1x2xX}), we are going to make use of the following proposition, which is stated in terms of macroscopic coordinates.
\begin{proposition}\label{2bodyscattering}
Suppose that $V$ is a non-negative, smooth, spherically symmetric, and compactly supported potential. Let $\fh=-2\Delta + V$ and denote by $1-\omega$ the solution to the zero energy scattering equation \[ \fh (1-\omega) = 0 \] with boundary condition $\omega (X) \to 0$ as $|X| \to \infty$. Moreover, let $\theta_L$ be defined as in (\ref{cutofftheta}) for some  $L>1$. Consider $\psi\in W^{3,1}(\mathbb{R}^3)\cap W^{3,\infty}(\mathbb{R}^3)$, with $\| \psi \| = 1$. Then, for $\Lambda \geq 1$, define $\psi_\Lambda$ as
\begin{equation}\label{psiL}
\psi_\Lambda(X) \; := \; \psi(X/\Lambda)
\end{equation}
with $\Lambda \gg L$. Define
\begin{equation}\label{F_Lambda_L}
F_{\Lambda,L}(T) := \int_{\mathbb{R}^3}\theta_L(X)\bigg|\nabla\,\frac{\big(e^{-i \fh T}\psi_\Lambda\big) (X)}{1-\omega(X)}\bigg|^2 \rd X\,.
\end{equation}
Then
\begin{equation}\label{Ft_decaying}
F_{\Lambda,L}(T) \; \leqslant \; C\,\tri\psi\tri^2\bigg(\frac{\:(\log T)^6}{\:T}L^3 \,+\,\frac{\,T^2+L^3}{\Lambda^2}\bigg)
\end{equation}
for all $T>0$ and $\Lambda \gg L$ sufficiently large. Here we used the notation
\begin{equation}
\tri\psi\tri \; := \; \|\psi\|_{W^{3,1}}+\|\psi\|_{W^{3,\infty}}.
\end{equation}
\end{proposition}

\medskip

\emph{Remark 1}. It is simple to check $F_{\Lambda,L}(0)\sim 1 + L^3/\Lambda^2$. Thus, Proposition \ref{2bodyscattering} states that  $F_{\Lambda,L}(T) \ll F_{\Lambda,L}(0)$, if $L^3 \ll \Lambda^2$, and times $T$ such that
\begin{equation}
1 \ll T \ll \Lambda\,.
\end{equation}
This fact can be interpreted as a sign for the formation of the local structure $1-\omega$ in $e^{-i \fh T}\psi_\Lambda$.

\medskip

\emph{Remark 2}. Taking formally $\Lambda = \infty$, (\ref{Ft_decaying}) describes the relaxation of a constant initial data towards the solution $1-\omega$ of the zero-energy scattering equation.

\medskip

\emph{Remark 3}. For $\Lambda<\infty$ the function $\psi_{\Lambda}$ can be thought of as cutting off $\psi_{\infty}\equiv 1$ at distances $\Lambda$. Since the energy hence the velocity of $\psi_{\Lambda}$ is of order 1 in $\Lambda$, a time of order $\Lambda$ is necessary for the effects of the cut-off to reach the window of size $L$: this explains why we can only prove that $F_{\Lambda,L}(T) \ll F_{\Lambda,L}(0)$ for times $T \ll \Lambda$.

\medskip

Applying Proposition \ref{2bodyscattering}, we immediately obtain the following bound for the integral (\ref{x1x2xX}).
\begin{corollary}\label{corollary1bodyanalysis}
Under the assumptions, and using the notation introduced in Proposition~\ref{prop:reduction},
\begin{equation}\label{est_two_body_piece}
\int_{\mathbb{R}^3\times\mathbb{R}^3}\!\rd \eta\,\rd x\,\theta_{2\ell}(x)\bigg|\nabla_{\!x}\,\frac{\,\big(\,e^{-i\fh_N t}\psi_\eta\big)(x)\,}{1-\omega_N(x)}  \,\bigg|^2 \; \leqslant \; C\,\bigg(\frac{\,(\log N^2t)^6}{\;N^2 t}N^2\ell^3 + N t^2 + \ell^3   \bigg)
\end{equation}
for all times $t>0$.
\end{corollary}

\begin{proof}
Changing coordinates to $X = N x$, we obtain, from Proposition \ref{2bodyscattering} (with $L = 2 N \ell$ and $\Lambda = N$),
\begin{equation}\label{eq:bdfN}
\int_{\bR^3} \rd x\,\theta_{2\ell}(x) \bigg|\nabla_{\!x}\,\frac{\,\big(\,e^{-i\fh_N t}\psi_\eta\big)(x)\,}{1-\omega_N(x)}  \,\bigg|^2  \leq C\,\tri\psi_\eta\tri^2 \bigg(\frac{\,(\log N^2t)^6}{\;N^2 t}N^2\ell^3 + N t^2 + \ell^3   \bigg) \, . \end{equation}
Since, by (\ref{superest}), $\tri \psi_\eta \tri \leq C \, \langle \eta \rangle^{-\alpha}$ for some $\alpha>3$, (\ref{est_two_body_piece}) follows from (\ref{eq:bdfN}), integrating over $\eta$.
\end{proof}

\emph{Remark}. By Corollary \ref{corollary1bodyanalysis}, the integral (\ref{x1x2xX}) can be shown to be smaller than its value at time $t= 0$, for all times $t$ in the interval $N^{-2} \ll t \ll N^{-1}$. Note that in the many body setting of Theorem \ref{N-thm}, on the other hand, we can only prove that $\cF_N (t) / \cF_N (0) \ll 1$ up to times $t\ll N^{-(2-\frac{1}{10})}$ (if $\ell \simeq N^{-1}$); this is due to the lack of control of the many body effects for larger times.

\medskip

Next, we prove Proposition \ref{2bodyscattering}. As we will see, the two main tools in the proof are Yajima's bounds on the wave operator $\Omega$ associated with the one-particle Hamiltonian $-\Delta + \frac{1}{2} V$ and a new generalized dispersive estimate for initial data which are slowly decreasing at infinity but have some regularity. The dispersive estimate is presented in Section \ref{sec:DispEst}. The definition and the most important properties of the wave operator are collected in Appendix \ref{app:waveop}.

\medskip

\begin{proof}[Proof of Proposition \ref{2bodyscattering}]
We start by splitting
\begin{equation}\label{FTsplit}
F_{\Lambda,L}(T) \; \leqslant \; 2\, F_{\Lambda,L}^{(1)}(T) \;+\; 2\, F_{\Lambda,L}^{(2)}(T)
\end{equation}
where
\begin{eqnarray}
F_{\Lambda,L}^{(1)}(T) & := & \int_{\mathbb{R}^3}\theta_L(X)\bigg|\nabla\,\frac{\big(e^{-i \fh T}\omega \psi_\Lambda\big)(X)}{1-\omega(X)}\bigg|^2 \rd X \, , \label{defF1}\\
F_{\Lambda,L}^{(2)}(T) & := & \int_{\mathbb{R}^3}\theta_L(X)\bigg|\nabla\,\frac{\big(e^{-i \fh T}(1-\omega)\psi_\Lambda\big)(X)}{1-\omega(X)}\bigg|^2 \rd X \, . \label{defF2}
\end{eqnarray}

\medskip

Let us first estimate $F_{\Lambda,L}^{(1)}$. Let $\Omega$ be the wave operator associated with $-\Delta_X+\frac{1}{2}V(X)$ (that is, our $\frac{1}{2}\fh$) as defined in Proposition \ref{prop:waveop}. Then, by the intertwining relation (\ref{eq:intertw}),
\begin{equation}
\begin{split}
F_{\Lambda,L}^{(1)}(T) \; & =
\;\int_{\mathbb{R}^3}\rd X\,\theta_L(X)\bigg|\nabla\,\frac{\,\big(\Omega\, e^{\,2i T\Delta}\,\Omega^*\omega\psi_\Lambda\big)(X)}{1-\omega(X)}\bigg|^2  \\
& \leqslant\; C\,\Big(\;\|\nabla\omega\|_{2}^2\,\big\|\,\Omega\, e^{\,2i T\Delta}\,\Omega^*\omega\psi_\Lambda\,\big\|_{\infty}^2
\;+\;L^{3}\big\|\,\nabla\Omega\, e^{\,2i T\Delta}\,\Omega^*\omega\psi_\Lambda \,\big\|_{\infty}^2\;\Big)\,.
\end{split}
\end{equation}
Here we have also used that $1-\omega(X)\geqslant\mathrm{const}>0$ (see Lemma~\ref{lemmaomega}). By Yajima's bound (\ref{yb}), we have
\begin{equation}
\big\|\,\Omega\, e^{\,2i T\Delta}\,\Omega^*\omega\psi_\Lambda\,\big\|_{\infty}^2  \leqslant C\,\big\|\,e^{\,2i T\Delta}\,\Omega^*\omega\psi_\Lambda\,\big\|_{\infty}^2
\end{equation}
and
\begin{equation}
\begin{split}
\big\|\,\nabla\Omega\, e^{\,2i T\Delta}\,\Omega^*\omega\psi_\Lambda \,\big\|_{\infty}^2
& \leqslant C\,\Big(\,\big\|\,e^{\,2i T\Delta}\,\Omega^*\omega\psi_\Lambda\,\big\|_{\infty}^2 + \big\|\,\nabla\,e^{\,2i T\Delta}\,\Omega^*\omega\psi_\Lambda\,\big\|_{\infty}^2   \,\Big)\,.
\end{split}
\end{equation}
Therefore, using (\ref{bound-on-nablaomega}) to bound $\| \nabla \omega \|_2$,
\begin{equation}\label{F_1sobolev}
F_{\Lambda,L}^{(1)}(T) \;\leqslant\;  C\,\Big(\,\big\|\,e^{\,2i T\Delta}\,\Omega^*\omega\psi_\Lambda\,\big\|_{\infty}^2 + L^3\big\|\,\nabla\,e^{\,2i T\Delta}\,\Omega^*\omega\psi_\Lambda\,\big\|_{\infty}^2   \,\Big)\,.
\end{equation}
{F}rom Proposition \ref{prop-Omegaomegapsilambda-dispersive} below, we get
\begin{equation}\label{F1final}
F_{\Lambda,L}^{(1)}(T) \;\leqslant\;  c_s\,\tri\psi\tri^2\,\frac{\,L^{3}} {\,T^{\,3/s}}\qquad\qquad\forall s\in (3,\infty]
\end{equation}
where $c_s\sim(s-3)^{-6}$ as $s\to 3^+$.

\medskip

Let us now estimate the term $F_{\Lambda,L}^{(2)}$ defined in (\ref{defF2}). We rewrite it conveniently by means of the operator
\begin{equation*}
\cL \; := \; -\Delta+2\,\frac{\nabla\omega}{1-\omega}\cdot\nabla
\end{equation*}
already introduced in (\ref{defcL}) (in microscopic variables). Analogously to (\ref{eq:evL}), we have
\begin{equation*}
e^{-i \fh T}\, (1-\omega)\phi\;=\;(1-\omega)e^{-2 i\cL T}\phi\, ,\qquad\;\forall\,\phi\in L^2(\mathbb{R}^3)
\end{equation*}
and
\begin{equation*}
\langle\phi,\cL\psi\rangle_\bullet = \langle\cL\phi,\psi\rangle_\bullet =  \langle\nabla\phi,\nabla\psi\rangle_\bullet
\end{equation*}
where bulleted scalar products here are in the Hilbert space $L^2\big(\mathbb{R}^3,(1-\omega(X))^2\rd X\big)$. On this space $\cL$ acts then as a selfadjoint operator. So
\begin{equation}\label{F2}
\begin{split}
F_{\Lambda,L}^{(2)}(T) \; & = \;\int_{\mathbb{R}^3}\rd X\,\theta_L(X)\big|\,\nabla e^{-2i \cL T} \psi_\Lambda(X)\,\big|^2 \\
& \leqslant \; \int_{\mathbb{R}^3}\rd X\,\theta_L(X)\big|\,\nabla\psi_\Lambda(X)\,\big|^2+\int_{\mathbb{R}^3}\rd X\,\big|\,\nabla (e^{-2i \cL T}-\mathbbm{1}) \psi_\Lambda(X)\,\big|^2\,.
\end{split}
\end{equation}
The first summand in the r.h.s.~of (\ref{F2}) is bounded as
\begin{equation}
\int_{\mathbb{R}^3}\rd X\,\theta_L(X)\big|\,\nabla\psi_\Lambda(X)\,\big|^2\, \; \leqslant \; C\,\frac{L^3}{\Lambda^2}\|\nabla\psi\|_\infty^2\; \leqslant \; C\,   \tri\psi\tri^2\frac{L^3}{\,\Lambda^2}\,.
\end{equation}
The second summand in the r.h.s.~of (\ref{F2}) is bounded by
\begin{equation}\label{F3TA}
\begin{split}
\!\!\!\!\!\!\!\int_{\mathbb{R}^3}\rd X\big|\,\nabla (e^{-2i \cL T}-\mathbbm{1}) \psi_\Lambda(X)\,\big|^2 \; & = \;\int_{\mathbb{R}^3}\,\bigg|\,\nabla\!
\int_0^T\rd S
\, e^{-2i \cL S}\cL \,\psi_\Lambda(X)\,\bigg|^2 \\
& \leqslant \; C\int_{\mathbb{R}^3}\rd X\,\big(1-\omega(X)\big)^2\,\bigg|\,\nabla\!
\int_0^T\rd S
\, e^{-2i \cL S}\cL \,\psi_\Lambda(X)\,\bigg|^2 \\
& \leqslant \; C\: \bigg\|\int_0^T\rd S
\, \nabla e^{-2i \cL S}\cL \,\psi_\Lambda\;\bigg\|^2_\bullet \\
& \leqslant \; C \,T^2\! \sup_{S\in[0,T]}\big\|
\, \nabla e^{-2i \cL S}\cL \,\psi_\Lambda \;\big\|^2_\bullet \\
& \leqslant \; C \,T^2\big\|\,\nabla\cL \,\psi_\Lambda\,\big\|^2_\bullet
\end{split}
\end{equation}
because, by (\ref{Lsimmetry}),
\begin{equation}
\begin{split}
\big\|
\, \nabla e^{-2i \cL S}\cL \,\psi_\Lambda \;\big\|^2_\bullet\; & = \; \langle\,\nabla e^{-2i \cL S}\cL \,\psi_\Lambda\,, \nabla e^{-2i \cL S}\cL \,\psi_\Lambda\, \rangle_\bullet \\
& =\; \langle\,e^{-2i \cL S}\cL \,\psi_\Lambda\,,\cL\, e^{-2i \cL S}\cL \,\psi_\Lambda\, \rangle_\bullet \\
& =\; \langle\,\cL \,\psi_\Lambda\,,\cL^2\,\psi_\Lambda\, \rangle_\bullet \\
& =\; \langle\,\nabla\cL \,\psi_\Lambda\,,\nabla\cL \,\psi_\Lambda\, \rangle_\bullet \\
& =\; \big\|\,\nabla\cL \,\psi_\Lambda\,\big\|^2_\bullet\;.
\end{split}
\end{equation}
In turn, $\big\|\,\nabla\cL\,\psi_\Lambda\,\big\|^2_\bullet\,$ is estimated as
\begin{equation}\label{nablaLchi}\begin{split}
\big\|\,\nabla\cL \,\psi_\Lambda\,\big\|^2_\bullet & \leqslant\;2\,\big\|\,\nabla^3\psi_\Lambda\,\big\|^2_\bullet + 2\,\Big\|\,\nabla\,\Big(\frac{\nabla\omega}{1-\omega}\cdot\nabla\psi_\Lambda\,\Big)\,\Big\|^2_\bullet \\
& \leqslant\;2\,\big\|\,\nabla^3\psi_\Lambda\,\big\|_2^2 + 2\,\Big\|\,\nabla\,\Big(\frac{\nabla\omega}{1-\omega}\cdot\nabla\psi_\Lambda\,\Big)\,\Big\|_2^2
\end{split}\end{equation}
where
\begin{equation}\label{nablaLchiI}
\big\|\,\nabla^3\psi_\Lambda\,\big\|_2^2\;=\;\frac{\|\,\nabla^3\psi\,\|_2^2}{\Lambda^3}\;\leqslant\;\frac{\;\tri\psi\tri^2}{\Lambda^3}
\end{equation}
and
\begin{equation}\label{nablaLchiII}\begin{split}
\Big\|\,\nabla\,\Big(\frac{\nabla\omega}{1-\omega}\cdot\nabla\psi_\Lambda\,\Big)\,\Big\|_2^2\; & \leqslant\;C\,\big(\,\big\|\,|\nabla^2\omega|\,|\nabla\psi_\Lambda|\,\big\|_2^2+\big\|\,|\nabla\omega|^2\,|\nabla\psi_\Lambda|\,\big\|_2^2+\big\|\,|\nabla\omega|\,|\nabla^2\psi_\Lambda|\,\big\|_2^2\,\big) \\
& \leqslant \; C\,\frac{\;\tri\psi\tri^2}{\Lambda^2}\,.
\end{split}\end{equation}
Thus,
\begin{equation}
\int_{\mathbb{R}^3}\rd X\big|\,\nabla (e^{-2i \cL T}-\mathbbm{1}) \psi_\Lambda(X)\,\big|^2 \; \leqslant \; C\, \frac{T^2}{\Lambda^2}\,\tri\psi\tri^2\,.
\end{equation}
Altogether, we find
\begin{equation}\label{F2final}
F_{\Lambda,L}^{(2)}(T)\;\leqslant\;C\,\tri\psi\tri^2 \,\frac{\,T^2+L^3}{\Lambda^2}\,.
\end{equation}

\medskip

Finally, plugging estimates (\ref{F1final}) for $F_{\Lambda,L}^{(1)}$ and (\ref{F2final}) for $F_{\Lambda,L}^{(2)}$ into (\ref{FTsplit}), one gets
\begin{equation}\label{almostfinal}
F_{\Lambda,L}(T) \; \leqslant \; C\,\tri\psi\tri^2\bigg(c_s\frac{L^3}{\:T^{\,3/s}} \,+\,\frac{\,T^2+L^3}{\Lambda^2}\bigg)\qquad\quad\forall s\in(3,\infty]\,.
\end{equation}
Optimizing in $s$ (since $c_s \simeq (s-3)^{-6}$ as $s \to 3^+$), we find
\begin{equation}
F_{\Lambda,L}(T) \; \leqslant \; C\,\tri\psi\tri^2\bigg(\frac{\:(\log T)^6}{\:T}L^3 \,+\,\frac{\,T^2+L^3}{\Lambda^2}\bigg)
\end{equation}
and (\ref{Ft_decaying}) is proved.
\end{proof}

\medskip

The following Proposition, which played an important role in the proof of Proposition~\ref{2bodyscattering} and was also used in Step 3 of the proof of Proposition \ref{prop:reduction}, provides an estimate for the decay (in $L^{\infty}$) of the evolution generated by the Hamiltonian $\fh = -2\Delta + V(x)$ on initial data decaying only as $|x|^{-1}$. It is based on new dispersive bounds for the free evolution which are presented in Section \ref{sec:DispEst}.

\begin{proposition}\label{prop-Omegaomegapsilambda-dispersive}
Under the same assumptions, and using the same notation as in Proposition \ref{2bodyscattering}, we have, for every $T\in\mathbb{R}$,
\begin{equation}\label{eq-Omegaomegapsilambda-dispersive}
\big\|\,e^{-i T\Delta}\,\Omega^*\omega\psi_\Lambda\,\big\|_{W^{1,\infty}} \; \leqslant \; c_s\,\frac{\tri\psi\tri}{\:T^{\frac{3}{2s}}}\qquad\qquad\forall s\in (3,+\infty]\, ,
\end{equation}
uniformly in $\Lambda$. Here $c_s\sim(s-3)^{-3}$ as $s\to 3^+$.
\end{proposition}

\emph{Remark}. For our purposes, the bound (\ref{eq-Omegaomegapsilambda-dispersive}) is better than the standard $L^1\to L^{\infty}$ dispersive estimate, because it is uniform in $\Lambda$ and we eventually want to let $\Lambda\to\infty$ (on the other hand, since $\omega$ only decays as $|x|^{-1}$, the standard $L^1\to L^{\infty}$ dispersive bound would diverge like $\Lambda^2$).

\begin{proof}
Let $\Omega$ be the wave operator associated with $-\Delta_X+\frac{1}{2}V(X)$. We will use the dispersive estimate
(\ref{eq:dispersive_estimate}) for $s\in(3,\infty]$, which gives
\begin{equation}\label{omegachi_dispersive}
\big\|\,e^{-i T\Delta}\,\Omega^*\omega\psi_\Lambda\,\big\|_{\infty}^2  \;\leqslant\;\frac{\,C}{\,T^{\,3/s}}\,\Big(\,
\big\|\,\Omega^*\omega\psi_\Lambda\,\big\|_{s}^2+\big\|\,\nabla\Omega^*\omega\psi_\Lambda\,\big\|_{\frac{3s}{s+3}}^2+\big\|\,\nabla^2\Omega^*\omega\psi_\Lambda\,\big\|_{\frac{3s}{2s+3}}^2
\,\Big)
\end{equation}
and
\begin{equation}\label{disp_for_nabla}
\begin{split}
\!\!\!\big\|\,\nabla e^{-i T\Delta}\,\Omega^*\omega\psi_\Lambda\,\big\|_{\infty}^2 \; & \leqslant \; \frac{C}{\,T^{\,3/s}}
\Big(\,\big\|\,\nabla\Omega^*\omega\psi_\Lambda\big\|_{s}^2   \,+\big\|\,\nabla^2\Omega^*\omega\psi_\Lambda\big\|_{\frac{3s}{s+3}}^2+\big\|\,\nabla^3\Omega^*\omega\psi_L\big\|_{\frac{3s}{2s+3}}^2 \, \Big)\,, 
\end{split}
\end{equation}
respectively. For convenience we set
\begin{equation}\label{v_z}
\begin{split}
v & := \frac{3s}{s+3}\in \Big(\,\frac{3}{2},3\,\Big] \quad \text{and} \quad
z := \frac{3s}{2s+3}\in\Big(1,\frac{3}{2}\,\Big]\,.
\end{split}
\end{equation}
We shall make use of the following bounds, which are a consequence of the properties of $\omega$ (see Lemma \ref{lemmaomega}), and are proven separately in Lemma \ref{lemma:norms-omega-psiL}:
\begin{eqnarray}
\|\nabla^m (\omega\psi_\Lambda)\|_p\! & \leqslant & \! c_{p,m}\,\tri\psi\tri\qquad\forall p\in[1,\infty]\,,\;\forall m\in\{0,1,2,3\}\;\textrm{ s.t. } p(m+1)>3 \label{estimate-nabla-m-omegapsiLambda} \\
\|\nabla^m (V \omega\psi_\Lambda) \|_p\! & \leqslant & \! C\,\tri\psi\tri\qquad\;\;\;\,\forall p\in[1,\infty]\,,\;\forall m\in\{0,1\}\,. \label{estimate-nabla-m-VomegapsiLambda}
\end{eqnarray}

\medskip

One sees that $\Omega^*\omega\psi_\Lambda\in L^s(\mathbb{R}^3)$ $\forall s>3$ because, by Yajima's bound (\ref{yb}) and by estimate (\ref{estimate-nabla-m-omegapsiLambda}),
\begin{eqnarray}\label{OmegaomegachiL}
\big\|\,\Omega^*\omega\psi_\Lambda\,\big\|_{s} \;\leqslant\;c_s\|\,\omega\psi_\Lambda\|_{s}\;\leqslant\;c_s\tri\psi\tri\,.
\end{eqnarray}
Here $c_s\sim(s-3)^{-1/3}$ for $s\to 3^+$.

\medskip

To treat $\|\,\nabla\Omega^*\omega\psi_\Lambda\|_v$ we use the property
\begin{equation}
\big\|\,\nabla f\,\big\|_{p} \; \leqslant \; c_p \,\big\|\,\sqrt{-\Delta\,}\,f\,\big\|_p\qquad\qquad\forall p\in (1,\infty)
\end{equation}
valid for every $f\in L^2$ such that $\sqrt{-\Delta\,}f\in L^2$ (see Lemma \ref{nablasqrtdelta}). This is the case for $f=\Omega^*\omega\psi_\Lambda$, due to Yajima's bound, although the corresponding norms are not bounded uniformly in $\Lambda$. In our application we need $p\in(\frac{3}{2},3]$, so we are never at the borderline case and $c_p$ remains uniformly bounded.  Thus, we get
\begin{equation}\label{treatment_of_nablaomegapsiL}
\begin{split}
\big\|\,\nabla\Omega^*\omega\psi_\Lambda\,\big\|_{v} \; & \leqslant\;c_v\,\big\|\,\sqrt{-\Delta\,}\,\Omega^*\omega\psi_\Lambda\,\big\|_{v}\;=\;c_v\,\big\|\,\Omega^*\sqrt{\fh\,}\,\omega\psi_\Lambda\,\big\|_{v} \\
& \leqslant\;c_v\,\big\|\,\sqrt{\fh\,}\,\omega\psi_\Lambda\,\big\|_{v}\;\leqslant\;c_v\,\int_0^{\infty}\frac{\rd k}{\sqrt{k\,}\,}\,\Big\|\, \frac{\mathbbm{1}}{k+\fh}\,\fh\,\omega\psi_\Lambda \,\Big\|_v \\
& \leqslant\;c_v\,\int_0^{\infty}\frac{\rd k}{\sqrt{k\,}\,}\,\Big\|\, \frac{\mathbbm{1}}{-\Delta + k\,}\,|\fh\,\omega\psi_\Lambda |\,\Big\|_v
\end{split}\end{equation}
where we have used intertwining relation (\ref{eq:intertw}), Yajima's bound (\ref{yb}), the fact that $\fh$ generates a positivity-preserving semigroup, and
\begin{equation}
\sqrt{\fh\,}=\frac{1}{\pi}\int_0^{\infty}\frac{\rd k}{\sqrt{k\,}\,}\,\frac{\mathbbm{1}}{k+\fh}\,\fh\,.
\end{equation}
We notice that estimates (\ref{estimate-nabla-m-omegapsiLambda}) and (\ref{estimate-nabla-m-VomegapsiLambda}) give
\begin{equation}\label{xia}
\|\, \fh\,\omega\psi_\Lambda\|_a \;\leqslant\; \|-\!\Delta\, (\omega\psi_\Lambda)\|_a \;+\;\|\,V\omega\psi_\Lambda\|_a \; \leqslant\; c_a\,\tri\psi\tri \qquad\quad a\in(1,\infty]\,.
\end{equation}
Thus, for any $b\in[1,3)\cap[1,v)$, Young's inequality gives
\begin{equation}\label{youngforkernel}
\begin{split}
\Big\|\, \frac{\mathbbm{1}}{-\Delta + k\,}\,|\fh\,\omega\psi_\Lambda |\,\Big\|_v  \; & = \;
\Big\|\, \frac{1}{4\pi}\inthree\rd Y \frac{\,e^{-\sqrt{k\,}\,|X-Y|}}{\,|X-Y|\,}\,\fh\,\omega\psi_\Lambda(Y) \,\Big\|_v  \\
& \leqslant \; c_{v,b}\,\Big\|\, \frac{\,e^{-\sqrt{k\,}\,|\,\cdot\,|}}{\,|\cdot|\,}\Big\|_b\,\|\,\fh\,\omega\psi_\Lambda\|_{(1+\frac{1}{v}-\frac{1}{b})^{-1}} \\
& = \; c_{v,b} \, k^{\frac{1}{2}-\frac{3}{2 b}}\|\,\fh\,\omega\psi_\Lambda\|_{(1+\frac{1}{v}-\frac{1}{b})^{-1}}\,.
\end{split}
\end{equation}
We use this bound to control the r.h.s. of (\ref{treatment_of_nablaomegapsiL}), applying it with two different values of $b$ in the region $|k| \leq 1$ and $|k| \geq 1$. For $v\in(\frac{3}{2},3)$, we find
\begin{equation}
\big\|\,\nabla\Omega^*\omega\psi_\Lambda\big\|_{v} \leqslant  c_v\,\Big(\|\,\fh\,\omega\psi_\Lambda\|_{\frac{3v^2}{12v-9-v^2}}+\|\,\fh\,\omega\psi_\Lambda\|_v\Big)\,. \label{nablaOmegaomegachiL-partial-v<3}
\end{equation}
Here $c_v\sim(v-\frac{3}{2})^2$ when $v\to\frac{3}{2}^+$
For $v=3$ we find, on the other hand,
\begin{equation}
\big\|\,\nabla\Omega^*\omega\psi_\Lambda\big\|_{3} \leqslant  \frac{C}{\mu(1-\mu)^{\frac{1}{3}}}\,\Big(\|\,\fh\,\omega\psi_\Lambda\|_{\frac{3}{2+\mu}}+\|\,\fh\,\omega\psi_\Lambda\|_3\Big) \label{nablaOmegaomegachiL-partial-v=3}
\end{equation}
for any $\mu\in(0,1)$.  Hence, by (\ref{xia}),
\begin{equation}\label{nablaOmegaomegachiL}
\begin{split}
\big\|\,\nabla\Omega^*\omega\psi_\Lambda\big\|_{\frac{3s}{3+s}} & \; \leqslant\; c_s\,\tri\psi\tri
\end{split}
\end{equation}
uniformly in $\Lambda$ and $\forall s\in(3,\infty]$, and $c_s\sim(s-3)^{-3}$ as $s\to 3^+$.

\medskip

To treat $\|\,\nabla^2\Omega^*\omega\psi_\Lambda\|_z$ we use the Calderon-Zygmund inequality
\begin{equation}\label{CZineq}
\|\nabla^2 f \|_p \; \leqslant \;  c_p \|\Delta f \|_p \qquad\qquad \forall p\in(1,\infty)
\end{equation}
valid for any compactly supported $f$ in $W^{2,p}(\mathbb{R}^3)$
(see \cite{Gilbarg-Trudinger}, Theorem 9.9).
The behaviour of the constant $c_p$ as $p\to 1^+$ can be computed from the constant in the Marcinkiewicz interpolation theorem (see \cite{Gilbarg-Trudinger}, Theorem 9.8) and one has $c_p\sim(p-1)^{-1}$.  Although $\supp(\Omega^*\omega\psi_\Lambda)$ is not compact, we can apply (\ref{CZineq}) to the compactly supported function
$\chi\subA\,\Omega^*\omega\psi_\Lambda$, where we have introduced the cut-off function
\begin{equation}
\chi\subA(Y) \; := \; \chi\Big(\frac{|Y|}{A}\Big)
\end{equation}
at the scale $A$, with $\chi$ defined as in (\ref{def_cutoff}). Then one has
\begin{equation}
\begin{split}
\big\|\,\nabla^2\Omega^*\omega\psi_\Lambda\big\|_{z} & \leqslant \; \big\|\,\nabla^2\chi\subA\,\Omega^*\omega\psi_\Lambda\big\|_{z} \; + \; \big\|\,\nabla^2 ( 1-\chi\subA\,)\Omega^*\omega\psi_\Lambda\big\|_{z} \\
& \leqslant\; c_z\,\big\|\,\Delta\,\chi\subA\,\Omega^*\omega\psi_\Lambda\big\|_{z} \; + \; \big\|\,\nabla^2 ( 1-\chi\subA\,)\Omega^*\omega\psi_\Lambda\big\|_{z} \\
& \leqslant\; c_z\,\big\|\,\Delta\,\Omega^*\omega\psi_\Lambda\big\|_{z}
\; + \; c_z\,\big\|\,\Delta (1-\chi\subA\,)\Omega^*\omega\psi_\Lambda\big\|_{z} \; + \; \big\|\,\nabla^2 ( 1-\chi\subA\,)\Omega^*\omega\psi_\Lambda\big\|_{z} \\
& \leqslant\;
c_z\,\big\|\,\Delta\,\Omega^*\omega\psi_\Lambda\big\|_{z} \; + \; c_z\,\big\|\,\nabla^2 ( 1-\chi\subA\,)\Omega^*\omega\psi_\Lambda\big\|_{z}
\end{split}
\end{equation}
and
\begin{equation}
\lim_{A\to +\infty }\big\|\,\nabla^2 ( 1-\chi\subA\,)\Omega^*\omega\psi_\Lambda\big\|_{z} = 0
\end{equation}
because
\begin{equation}\begin{split}
\!\!\!\!\!\!\!\big\|\,\nabla^2 ( 1-\chi\subA\,)\Omega^*\omega\psi_\Lambda\big\|_{z} & \; \leqslant \; \big\|\,(\nabla^2\chi\subA\,)\Omega^*\omega\psi_\Lambda\big\|_{z} +  2\,\big\|\,(\nabla\chi\subA\,)\cdot\nabla\Omega^*\omega\psi_\Lambda\big\|_{z} +
\big\|\,( 1-\chi\subA\,)\nabla^2 \Omega^*\omega\psi_\Lambda\big\|_{z} \\
& \; \leqslant \; C\,\frac{\|\Omega^*\omega\psi_\Lambda\|_z}{A^2} \; + \; C\,\frac{\|\nabla\Omega^*\omega\psi_\Lambda\|_z}{A} \; + \big\|\,( 1-\chi\subA\,)\nabla^2 \Omega^*\omega\psi_\Lambda\big\|_{z} \\
& \;\; \xrightarrow[]{\;A\to +\infty\;}\, 0\,.
\end{split}\end{equation}
Last summand above, in particular, vanishes as $A\to+\infty$ by dominated convergence: in fact we have $( 1-\chi\subA\,)\nabla^2\Omega^*\omega\psi_\Lambda\to 0$ pointwise and is in $L^z(\mathbb{R}^3)$ uniformly in $A$ because
\begin{equation}
\big\|\,( 1-\chi\subA\,)\nabla^2 \Omega^*\omega\psi_\Lambda\big\|_{z} \; \leqslant \; \big\|\,\nabla^2 \Omega^*\omega\psi_\Lambda\big\|_{z} \;
\leqslant \; \big\|\,\Omega^*\omega\psi_\Lambda\big\|_{W^{2,z}} \;
\leqslant \; c_z\, \big\|\,\omega\psi_\Lambda\big\|_{W^{2,z}} < \infty\,.
\end{equation}
Thus,
\begin{equation}\label{nablasquarewithdelta}
\big\|\,\nabla^2\Omega^*\omega\psi_\Lambda\big\|_{z} \; \leqslant \; c_z\, \big\|\,\Delta\Omega^*\omega\psi_\Lambda\big\|_{z}\qquad\qquad\forall z\in(1,+\infty)\,.
\end{equation}
Using (\ref{nablasquarewithdelta}), the intertwining relation (\ref{eq:intertw}), and  Yajima's bound (\ref{yb}), one has (with $z=\frac{3s}{2s+3}$, $s>3$)
\begin{equation}\begin{split}\label{nabla2OmegaomegachiL}
\big\|\,\nabla^2\Omega^*\omega\psi_\Lambda\big\|_{z}  & \leqslant \; c_z\, \big\|\,\Delta\Omega^*\omega\psi_\Lambda\big\|_{z} \; \leqslant\;c_z\, \big\|\,\Omega^*\fh\,\omega\psi_\Lambda\big\|_{z} \\
& \leqslant \; c_z\, \big(\,\|\,\Delta\,(\omega\psi_\Lambda)\|_{z}+\big\|\,V\omega\psi_\Lambda\big\|_{z}\big) \\
& \leqslant \; c_s\,\tri\psi\tri
\end{split}\end{equation}
uniformly in $\Lambda$ and $\forall s\in(3,\infty]$, where (\ref{estimate-nabla-m-omegapsiLambda}) and (\ref{estimate-nabla-m-VomegapsiLambda}) have been used in the last line. Following the blow-up of the various constants, we see that in (\ref{nabla2OmegaomegachiL}) $c_s\sim(s-3)^{-2}$ as $s\to 3^+$.

\medskip

{F}rom (\ref{OmegaomegachiL}), (\ref{nablaOmegaomegachiL}), and (\ref{nabla2OmegaomegachiL}), the dispersive estimate (\ref{omegachi_dispersive}) takes the form
\begin{equation}\label{estimate_ft}
\big\|\,e^{-iT\Delta}\,\Omega^*\omega\psi_\Lambda\big\|_{\infty}^2  \;\leqslant\;c_s\frac{\,\tri\psi\tri^2}{\,T^{3/s}}\qquad\qquad\forall s\in (3,\infty]
\end{equation}
with $c_s\sim (s-3)^{-6}$ as $s\to 3^+$.

\medskip

Next we treat the r.h.s.~of (\ref{disp_for_nabla}) in analogy to what we have done so far for the r.h.s.~of (\ref{omegachi_dispersive}). The first term on the r.h.s. of (\ref{disp_for_nabla}) is bounded as
\begin{equation}\label{part_i1}
\begin{split}
\|\,\nabla\Omega^*\omega\psi_\Lambda\|_{s}\; & \leqslant\;C\,\big\|\,\Omega^*\omega\psi_\Lambda\big\|_{W^{1,s}}\;\leqslant\;c_s\,\big\|\,\omega\psi_\Lambda\big\|_{W^{1,s}} \; \leqslant\; c_s\,\tri\psi\tri
\end{split}
\end{equation}
for any $s>3$ and uniformly in $\Lambda$: the second inequality above follows from Yajima's bound (\ref{yb}) and the last one follows from estimates (\ref{estimate-nabla-m-omegapsiLambda}) and (\ref{estimate-nabla-m-VomegapsiLambda}).
The second summand on the r.h.s.~of (\ref{disp_for_nabla}) is estimated analogously to (\ref{nabla2OmegaomegachiL}) and gives
\begin{equation}\label{part_i2}
\big\|\,\nabla^2\Omega^*\omega\psi_\Lambda\big\|_{v} \;
\leqslant \; c_s\,\tri\psi\tri
\end{equation}
uniformly in $\Lambda$.
Finally, the third term on the r.h.s. of (\ref{disp_for_nabla}) can be bounded, similarly to  (\ref{nablasquarewithdelta}), by
\begin{equation}
\big\|\,\nabla^3\Omega^*\omega\psi_\Lambda\big\|_{z}\;\leqslant\;\sum_{i=1}^3\big\|\,\nabla^2\partial_i\,\Omega^*\omega\psi_\Lambda\big\|_{z} \; \leqslant \; c_z\sum_{i=1}^3 \big\|\,\Delta\partial_i\,\Omega^*\omega\psi_\Lambda\big\|_{z}\,.
\end{equation}
Hence, by means of the intertwining relation (\ref{eq:intertw}), of Yajima's bound (\ref{yb}), and of the estimates (\ref{estimate-nabla-m-omegapsiLambda}) and (\ref{estimate-nabla-m-VomegapsiLambda}), we find
\begin{equation}\label{part_i3}\begin{split}
\big\|\,\Delta\partial_i\Omega^*\omega\psi_\Lambda\big\|_{z} & = \; c_z\,\big\|\,\partial_i\Omega^*\fh\,\omega\psi_\Lambda\big\|_{z}  \\ & \leqslant \; c_z\,\big\|\,\fh\,\omega\psi_\Lambda\big\|_{W^{1,z}} \\
& \leqslant \; c_z\,\big(\,\big\|\,\Delta\, (\omega\psi_\Lambda) \big\|_{z}  + \big\|\,V \omega\psi_\Lambda\big\|_{z} + \big\|\,\nabla^3 (\omega\psi_\Lambda)\big\|_{z} + \big\|\,\nabla (V \omega\psi_\Lambda ) \big\|_{z}
\big) \\
& \leqslant \; c_s\,\tri\psi\tri
\end{split}\end{equation}
uniformly in $\Lambda$. Thus, from (\ref{part_i1}), (\ref{part_i2}), and (\ref{part_i3}), the dispersive estimate (\ref{disp_for_nabla}) takes the form
\begin{equation}\label{estimate_nabla_ft}
\big\|\,\nabla e^{-i\,T\Delta}\,\Omega^*\omega\psi_\Lambda\big\|_{\infty}^2
\;\leqslant\;c_s\frac{\,\tri\psi\tri^2}{\:T^{3/s}}\qquad\qquad\forall s\in (3,\infty]\,.
\end{equation}
Following the blow-up of the various constants, we see that in (\ref{estimate_nabla_ft}) $c_s\sim(s-3)^{-4}$ as $s\to 3^+$.
\end{proof}

\medskip

In the following two lemmas, we prove estimates which were used in the proof of Proposition
\ref{prop-Omegaomegapsilambda-dispersive}.

\medskip

\begin{lemma}\label{lemma:norms-omega-psiL}
Let $V$, $\omega$, $\psi$, and $\psi_\Lambda$ be as in the hypothesis of Proposition \ref{2bodyscattering}. Then, for any $m\in\mathbb{N}$ and any $p\in[1,\infty]$ such that $p(m+1)>3$, there exists a constant $c_{p,m}$ such that
\begin{equation}\label{nabla-m-omega-psiLambda}
\big\|\nabla^m (\omega\psi_\Lambda) \big\|_p \; \leqslant \; c_{p,m}\,\tri\psi\tri\,.
\end{equation}
Here $c_{p,m}$ blows up as $(p\,(m+1)-3)^{-\frac{1}{p}}$ as $p\to(\frac{3}{m+1})^+$.
Moreover, for any $m=0,1$ and any $p\in[1,\infty]$, one has
\begin{equation}\label{nablaVomegapsiLambda}
\big\|\nabla^m (V\omega\psi_\Lambda)\big\|_p \; \leqslant \; C\,\tri\psi\tri
\end{equation}
for some constant $C$ depending on $V$.
\end{lemma}

\begin{proof}
We recall from Appendix \ref{app1bodyscatt} that $\omega$ satisfies the bounds
\begin{equation}\label{bound_on_omega}
\|\nabla^n\omega\|_q \; < \; \infty \qquad \forall q\in[1,\infty]\,,\;\forall n\in\mathbb{N}\,\textrm{ s.t. } q(n+1)>3
\end{equation}
with
\begin{equation}
\|\nabla^n\omega\|_q \; \sim \; \big(q(n+1)-3\big)^{-\frac{1}{q}}\qquad\mathrm{as}\;q\to\!\!\begin{array}{l}(\frac{3}{n+1})^+\end{array}\!.
\end{equation}
Moreover, $\psi_\Lambda$ satisfies the scaling
\begin{equation}\label{nablapsiLbyscaling}
\|\nabla^\nu\psi_\Lambda\|_{q} \;=\;\Lambda^{\frac{3}{\,q}-\nu}\,\|\nabla^\nu\psi\|_{q}\,.
\end{equation}
Pick $m\in\mathbb{N}$ and $p\in[1,\infty]$ such that $p(m+1)>3$: then, keeping into account (\ref{bound_on_omega}) and (\ref{nablapsiLbyscaling}), one obtains by H\"older inequality
\begin{equation}
\big\|\nabla^m (\omega\psi_\Lambda) \big\|_p \; \leqslant \; c_{p,m}\sum_{\nu=0}^m\big\|\nabla^\nu \psi\big\|_{\frac{p(m+1)}{\nu}}
\end{equation}
for some constants $c_{p,m}$ blowing up as $(p\,(m+1)-3)^{-\frac{1}{p}}$ when $p\to(\frac{3}{m+1})^+$. (When $\nu=0$ it is understood that $\frac{p(m+1)}{\nu}=\infty$). By interpolation,
\begin{equation}
\big\|\nabla^m (\omega\psi_\Lambda)\big\|_p \; \leqslant \; c_{p,m}\sum_{\nu=0}^3\big\|\nabla^\nu \psi\big\|_{\frac{p(m+1)}{\nu}}\;\leqslant\;c_{p,m}\,\big(\,\|\psi\|_{W^{3,1}}+\|\psi\|_{W^{3,\infty}}\big)
\end{equation}
that is, we obtain (\ref{nabla-m-omega-psiLambda}).
Eq. (\ref{nablaVomegapsiLambda}) can be proven similarly.
\end{proof}

\medskip

\begin{lemma}\label{nablasqrtdelta}
Let $f\in L^2(\mathbb{R}^3)$ such that $\sqrt{-\Delta\,}f\in L^2(\mathbb{R}^3)$. Then
\begin{equation}
\big\|\,\nabla f\,\big\|_{p} \; \leqslant \; c_p \,\big\|\,\sqrt{-\Delta\,}\,f\,\big\|_p
\end{equation}
for any $p\in(1,+\infty)$.
\end{lemma}

\begin{proof}
The boundedness of the singular integral operator (see, e.g., \cite{Stein-HarmonicAnalysis3}, Chapter III, Theorem 4, and \cite{TaylorPDE3}, Chapter 13, Theorem 5.1) implies that
\begin{equation}
\Big\|\,\nabla\frac{1}{\sqrt{-\Delta\,}}\,g \,\Big\|_p \; \leqslant \; C \|g\|_p\,.
\end{equation}
Let $g:=\sqrt{-\Delta\,}f$. Then $g\in L^2(\mathbb{R}^n)$. On $L^2(\mathbb{R}^n)$ the operator $\sqrt{-\Delta\,}$ has a trivial kernel, therefore
\begin{equation}
\frac{1}{\sqrt{-\Delta\,}}\,\sqrt{-\Delta\,}f \; = \; f\,.
\end{equation}
Then
\begin{equation}
\big\|\,\nabla f\,\big\|_{p} \; \leqslant \; c_p \,\big\|\,\sqrt{-\Delta\,}\,f\,\big\|_p\,.
\end{equation}
\end{proof}

\section{Dispersive estimate for regular, slowly decaying initial data}\label{sec:DispEst}

The standard dispersive estimate
\begin{equation}\label{standard_dispersive_estimate}
\big\|\,e^{i\Delta t}f\,\big\|_q\;\leqslant\;\frac{C}{\,t^{\,3(\frac{1}{s}-\frac{1}{2})}}\,\|f\|_s\qquad\begin{array}{rl}
s\!\!\!\! & \in [1,2] \\
q=\frac{s}{s-1}\!\!\!\! & \in [2,+\infty]
\end{array}
\end{equation}
for the free Schr\"odinger evolution is not suited for functions that decay slowly at infinity. In this section we prove a dispersive estimate which holds for $f\in L^s$ for any $s\in[\frac{3}{2},\infty]$, if additionally some $L^p$ bound is known on the derivatives of $f$.

\medskip

\begin{proposition}\label{modified_dispersive}
Let $s\in[\frac{3}{2},\infty]$, $q\in[\max\{s,3\},\infty]$, and $r\in[1,\frac{3q}{3+2q}]$.  Let $f\in L^s(\mathbb{R}^3)$ such that
\begin{equation}
\begin{split}
\nabla f & \in L^{\frac{3s}{s+3}}(\mathbb{R}^3) \\
\nabla^2 f & \in L^{r}(\mathbb{R}^3) .
\end{split}
\end{equation}
Then $e^{i\Delta t}f\in L^q(\mathbb{R}^3)$ and \begin{equation}\label{modif_dispersive_estimate}
\big\|\,e^{i\Delta t}f\,\big\|_q\leqslant \frac{\,C}{t^{\frac{3}{2}(\frac{1}{s}-\frac{1}{q})}}
\Big(\;\|f\|_s+\|\nabla f\|_{\frac{3s}{s+3}}
\Big) + \frac{\,C}{t^{\,\frac{3}{2}(\frac{1}{r}-\frac{1}{q})-1}}\|\nabla^2 f\|_{r}
\end{equation}
for some constant $C$ which is independent of $s,q,r$.
\end{proposition}

\medskip

\emph{Remark}. We use this estimate in the proof of Proposition \ref{prop-Omegaomegapsilambda-dispersive} and of Proposition \ref{prop:t-x-boundedness}, with $q=\infty$, $s\in[3,\infty]$ and $r=\frac{3s}{3+2s}$. In this case, (\ref{modif_dispersive_estimate}) reads
\begin{equation}\label{eq:dispersive_estimate}
\big\|\,e^{i\Delta t}f\,\big\|_\infty\leqslant \frac{\,C}{\,t^{\frac{3}{2s}}}
\Big(\;\|f\|_s+\|\nabla f\|_{\frac{3s}{s+3}}
+\|\nabla^2 f\|_{\frac{3s}{2s+3}}\Big)\,.
\end{equation}

\medskip

\begin{proof}[Proof of Proposition \ref{modified_dispersive}]
It is enough to prove (\ref{modif_dispersive_estimate}) for $f\in C^{\infty}_0(\mathbb{R}^3)$; then the estimate can be extended by a density argument.

\medskip

For $q \geq 1$, we have
\begin{equation}\label{ft_kernel}
\big\|\,e^{i\Delta t}f\,\big\|_q  = \frac{1}{(4\pi t)^{3/2}}\bigg(\inthree \rd x\,\left|\,\inthree \rd y \:e^{\frac{i|x-y|^2}{4t}} f(y)\,\right|^q\bigg)^{1/q}\,.
\end{equation}
We split the above integral for small and large values of $|x-y|$: in the latter regime integration by parts will provide the necessary decay at infinity. We introduce $R>0$ and we define the smooth cutoff function
\begin{equation}
\theta_R(x) \, := \, \chi\Big(\frac{\,|x|\,}{R}\Big)
\end{equation}
with $\chi$ defined in (\ref{def_cutoff}).
The following scaling properties of $\theta_R$ will be needed:
\begin{eqnarray}
\left\|\frac{\theta_R}{\,|\cdot|^m}  \right\|_p & = & A_{m,p}\,R^{\,\frac{3}{p}-m}\qquad\textrm{if }mp<3 \label{1-tau}\\
\left\|\frac{1-\theta_R}{\;|\cdot|^m}  \right\|_p & = & B_{m,p}\,R^{\,\frac{3}{p}-m}\qquad\textrm{if }mp>3 \label{tau} \\
\left\|\frac{\;\nabla\theta_R}{\;|\cdot|^m}  \right\|_p & = &
C_{m,p}\,R^{\,\frac{3}{p}-m-1}\quad\,\forall\,p\geqslant 1\,,\forall\,m\in\mathbb{R}\,. \label{gradtau}
\end{eqnarray}

\medskip

Inserting the cut-off in (\ref{ft_kernel}), we find
\begin{equation}\label{ftI+II}
\begin{split}
\big\|\,e^{i\Delta t}f\,\big\|_q & \leqslant\frac{\,C}{\,t^{3/2}}
\bigg(\inthree \rd x\,\left|\,\inthree \rd y\, \theta_R(x-y)\:e^{\frac{i|x-y|^2}{4t}}\,f(y)\,\right|^q\bigg)^{\frac{1}{q}} \\
& \quad +\frac{\,C}{\,t^{3/2}}
\bigg(\inthree \rd x\,\left|\,\inthree \rd y\, \big(1-\theta_R(x-y)\big)\:e^{\frac{i|x-y|^2}{4t}}\,f(y)\,\right|^q\bigg)^{\frac{1}{q}} \\
& \equiv (I)+(I')
\end{split}
\end{equation}
where first summand in the r.h.s.~is immediately estimated by Young's inequality as
\begin{equation}
(I)\;\leqslant\;\frac{\,c_{q,s}}{\,t^{3/2}}\,\|f\|_s\|\theta_R\|_{(1+\frac{1}{q}-\frac{1}{s})^{-1}}\;\leqslant\;c_{q,s}\frac{\,R^{\,3(1+\frac{1}{q}-\frac{1}{s})}}{t^{3/2}}\|f\|_s\,\qquad 1\leqslant s\leqslant q\,.
\end{equation}

\medskip

To estimate the summand $(I')$, we use
\begin{equation}\label{nablaexp}
e^{\frac{i|x-y|^2}{4t}} = 2it\frac{x-y}{\,|x-y|^2}\cdot\nabla_{\!y} \,e^{\frac{i|x-y|^2}{4t}}
\end{equation}
to write
\begin{equation}
(I')\;=\; \frac{\,C}{\,t^{1/2}}
\bigg(\inthree \rd x\,\left|\,\inthree \rd y\, \big(1-\theta_R(x-y)\big)\,f(y)\,\frac{x-y}{\,|x-y|^2}\cdot\nabla_{\!y}\,e^{\frac{i|x-y|^2}{4t}}\,\right|^q\bigg)^{\frac{1}{q}} \,.
\end{equation}
Then, integration by parts and since
\begin{equation}\label{divergencex-y}
\nabla_{\!y}\cdot\frac{x-y}{\,|x-y|^2}=-\frac{1}{\,|x-y|^2}\,,
\end{equation}
we bound $(I')$ as
\begin{equation}\label{I_2}
\begin{split}
(I') \; & \leqslant \; \frac{\,C}{\,t^{1/2}}
\bigg(\inthree \rd x\,\left|\,\inthree \rd y\, f(y)\,e^{\frac{i|x-y|^2}{4t}}\,\nabla\theta_R(x-y)\cdot\frac{x-y}{\,|x-y|^2}\,\right|^q\bigg)^{\frac{1}{q}} \\
& \qquad +\; \frac{\,C}{\,t^{1/2}}
\bigg(\inthree \rd x\,\left|\,\inthree \rd y\, \big(1-\theta_R(x-y)\big)\,\frac{f(y)}{\,|x-y|^2}\,e^{\frac{i|x-y|^2}{4t}}\,\right|^q\bigg)^{\frac{1}{q}}  \\
& \qquad +\; \frac{\,C}{\,t^{1/2}}
\bigg(\inthree \rd x\,\left|\,\inthree \rd y\,\big(1-\theta_R(x-y)\big) \,
e^{\frac{i|x-y|^2}{4t}}\,\frac{x-y}{\,|x-y|^2}\cdot\nabla f(y)\,\right|^q\bigg)^{\frac{1}{q}} \\
& \equiv \; (I\!I) + (I\!I\!I) + (IV)\,.
\end{split}
\end{equation}
The term $(I\!I)$ is estimated by Young's inequality and (\ref{gradtau}):
\begin{equation}
(I\!I) \; \leqslant \; \frac{\,c_{q,s}}{\,t^{1/2}}\,\|f\|_s\,\bigg\|\,\frac{\,\nabla\theta_R\,}{|\cdot|}\,\bigg\|_{(1+\frac{1}{q}-\frac{1}{s})^{-1}}\;=\;c_{q,s}\frac{\,R^{\,1+\frac{3}{q}-\frac{3}{s}}}{t^{1/2}}\|f\|_s\qquad 1\leqslant s\leqslant q\,.
\end{equation}

\medskip

To control the summands $(\text{III}), (\text{IV})$, we integrate by parts once more. {F}rom  (\ref{nablaexp}), (\ref{divergencex-y}), and
\begin{equation}\label{nabla1overy}
\nabla_{\!y}\frac{1}{\,|x-y|^2}=2\frac{x-y}{\,|x-y|^4}\,,
\end{equation}
we find
\begin{equation}\label{2nd_int_parts}
\begin{split}
\!\!\!\!\inthree\rd y\,\big(1-\theta_R(x-y)\big)\frac{f(y)}{\,|x-y|^2}\,e^{\frac{i|x-y|^2}{4t}} \; = &\; -2it\,\bigg\{\,\inthree \rd y\,\big(1-\theta_R(x-y)\big)\,\frac{f(y)}{\,|x-y|^4}\,e^{\frac{i|x-y|^2}{4t}}   \\
& + \inthree\rd y\, \big(1-\theta_R(x-y)\big)\,e^{\frac{i|x-y|^2}{4t}}\frac{x-y}{\,|x-y|^4}\cdot\nabla f(y) \\
& - \inthree\rd y\,\big(\nabla_{\!y}\theta_R(x-y)\big)\cdot\frac{x-y}{\,|x-y|^4}\,e^{\frac{i|x-y|^2}{4t}}f(y)\,\bigg\}\,.
\end{split}
\end{equation}
Therefore
\begin{equation}
\begin{split}
(I\!I\!I) \; & \leqslant \; C\,t^{\frac{1}{2}}
\bigg(\inthree \rd x\,\left|\,\inthree \rd y\, \big(1-\theta_R(x-y)\big)\,\frac{f(y)}{\,|x-y|^4}\,e^{\frac{i|x-y|^2}{4t}}\,\right|^q\bigg)^{\frac{1}{q}} \\
& \qquad +\; C\,t^{\frac{1}{2}}
\bigg(\inthree \rd x\,\left|\,\inthree \rd y\, \big(1-\theta_R(x-y)\big)\,e^{\frac{i|x-y|^2}{4t}}\,\frac{x-y}{\,|x-y|^4}\cdot\nabla f(y)\,\right|^q\bigg)^{\frac{1}{q}} \\
& \qquad +\; C\,t^{\frac{1}{2}}
\bigg(\inthree \rd x\,\left|\,\inthree \rd y\,e^{\frac{i|x-y|^2}{4t}}\, f(y)\,\nabla\theta_R(x-y)\cdot\frac{x-y}{\,|x-y|^4}\,\right|^q\bigg)^{\frac{1}{q}} \\
& \leqslant\; c_{q,s}\,t^{\frac{1}{2}}\,\bigg(\:\|f\|_s\,\left\|\,\frac{1-\theta_R}{|\cdot|^4}\,\right\|_{(1+\frac{1}{q}-\frac{1}{s})^{-1}} +\;\|\nabla f\|_{\frac{3s}{3+s}}\left\|\,\frac{1-\theta_R}{|\cdot|^3}\,\right\|_{(\frac{2}{3}+\frac{1}{q}-\frac{1}{s})^{-1}} \\
& \qquad\qquad\qquad +\; \|f\|_s\,\left\|\,\frac{\,\nabla\theta_R\,}{|\cdot|^3}\,\right\|_{(1+\frac{1}{q}-\frac{1}{s})^{-1}}\bigg)\\
& = \; \;c_{q,s}\,\frac{t^{\frac{1}{2}}}{\,R^{\,1+\frac{3}{s}-\frac{3}{q}}}\,\Big(\:\|f\|_s + \big\|\nabla f\big\|_{\frac{3s}{3+s}}\, \Big) \qquad\qquad\qquad \begin{array}{r}\frac{3}{2}\end{array}\!\!\!\leqslant s \leqslant q
\end{split}
\end{equation}
where Young's inequality and (\ref{tau}) and (\ref{gradtau}) have been used.

\medskip

Analogously, to handle the term (\text{IV}), we  performs a second integration by parts. We observe that
\begin{equation}\label{2nd_int_partsII}
\begin{split}
\inthree \rd y \:&\big(1-\theta_R(x-y)\big)\, e^{\frac{i|x-y|^2}{4t}}\frac{x-y}{\,|x-y|^2}\cdot\nabla f(y)\;= \\
& = 2it\,\left\{\inthree\rd y\, \right.\big(1-\theta_R(x-y)\big)\,e^{\frac{i|x-y|^2}{4t}}\:\frac{x-y}{\,|x-y|^4}\cdot\nabla  f(y) \\
& \qquad\qquad\quad -\inthree\rd y \,\big(1-\theta_R(x-y)\big)\,e^{\frac{i|x-y|^2}{4t}}\:\frac{x-y}{\,|x-y|^2}\cdot\nabla_{\!y}\left(\frac{x-y}{\,|x-y|^2}\cdot\nabla f(y)\right) \\
& \left.\qquad\qquad\quad +\inthree\rd y\,e^{\frac{i|x-y|^2}{4t}}\left(\frac{x-y}{\,|x-y|^2}\cdot\nabla f(y)\right)\frac{x-y}{\,|x-y|^2}\cdot\nabla\theta_R(x-y)\right\}\,.
\end{split}
\end{equation}
Therefore
\begin{equation}
\begin{split}
(IV) \; & \leqslant \; C\,t^{\frac{1}{2}}\bigg(\inthree \rd x\,\left(\,\inthree \rd y\,\big(1-\theta_R(x-y)\big)\frac{|\nabla f(y)|}{\,|x-y|^3}\,\right)^q\bigg)^{\frac{1}{q}} \\
& \quad +C\,t^{\frac{1}{2}}\bigg(\inthree \rd x\,\left(\,\inthree\rd y\, \big(1-\theta_R(x-y)\big)\frac{|\nabla^2 f(y)|}{\,|x-y|^2}\,\right)^q\bigg)^{\frac{1}{q}} \\
& \quad +\, C\,t^{\frac{1}{2}}\bigg(\inthree \rd x\,\left(\,\inthree\rd y\,\big|\nabla\theta_R(x-y)\big|\,\frac{|\nabla f(y)|}{\,|x-y|^2}\,\right)^q\bigg)^{\frac{1}{q}} \\
& \leqslant \; c_{q,s}\,t^{\frac{1}{2}}\,\|\nabla f\|_{\frac{3s}{3+s}}\left\|\,\frac{1-\theta_R}{|\cdot|^3}\,\right\|_{(\frac{2}{3}+\frac{1}{q}-\frac{1}{s})^{-1}} + \; c_{q,r}\,t^{\frac{1}{2}}\,\big\|\nabla^2 f\big\|_{r}\,\left\|\, \frac{1-\theta_R}{|\cdot|^2}\,\right\|_{(1+\frac{1}{q}-\frac{1}{r})^{-1}} \\
& \qquad\qquad\qquad + \; c_{q,s}\,t^{\frac{1}{2}}\,\big\|\nabla f\big\|_{\frac{3s}{3+s}}\,\left\|\, \frac{\nabla\theta_R}{|\cdot|^2}\,\right\|_{(\frac{2}{3}+\frac{1}{q}-\frac{1}{s})^{-1}} \\
& = \; c_{q,s}\,\frac{t^{\frac{1}{2}}}{\,R^{\,1+\frac{3}{s}-\frac{3}{q}}}\,\big\|\nabla f\big\|_{\frac{3s}{3+s}} + c_{q,r}\,\frac{t^{\frac{1}{2}}}{\,R^{\,\frac{3}{r}-\frac{3}{q}-1}}\big\|\nabla^2 f\big\|_{r}
\end{split}
\end{equation}
for $3/2\leqslant s \leqslant q \leqslant \infty\,, \;q\neq 3/2$ and $1 \leqslant r < 3q/(3+q)$. Here Young's inequality and the scaling properties (\ref{tau}) and (\ref{gradtau}) have been used.

\medskip

Summarizing,
\begin{equation}\label{est_all_indeces}
\begin{split}
\big\|\,e^{i\Delta t}f\,\big\|_q  & \leqslant\; (I)+(I\!I)+(I\!I\!I)+(IV)\\
& \leqslant\; c_{q,s}\bigg(\frac{\,R^{\,3+\frac{3}{q}-\frac{3}{s}}}{t^{\frac{3}{2}}}+\frac{\,R^{\,1+\frac{3}{q}-\frac{3}{s}}}{t^{\frac{1}{2}}}+\frac{t^{\frac{1}{2}}}{\,R^{\,1+\frac{3}{s}-\frac{3}{q}}}\bigg)\|f\|_s \;+\; c_{q,s}\,\frac{t^{\frac{1}{2}}}{\,R^{\,1+\frac{3}{s}-\frac{3}{q}}}\,\big\|\nabla f\big\|_{\frac{3s}{3+s}} \\
& \qquad +\; c_{q,r}\,\frac{t^{\frac{1}{2}}}{\,R^{\,\frac{3}{r}-\frac{3}{q}-1}}\big\|\nabla^2 f\big\|_{r}
\end{split}
\end{equation}
for any $R>0$, $3/2 \leqslant s \leqslant q\leqslant\infty$, $q\neq 3/2$, $1\leqslant r < 3q/(3+q)$. Optimizing the choice of $R$ leads to $R=\sqrt{t\,}$, so that (\ref{est_all_indeces}) reads
\begin{equation}\label{almostf}
\big\|\,e^{i\Delta t}f\,\big\|_q\;\leqslant\; \frac{\,c_{q,s}}{\,t^{\,\frac{3}{2}(\frac{1}{s}-\frac{1}{q})}}
\Big(\;\|f\|_s+\|\nabla f\|_{\frac{3s}{s+3}}
\Big) + \frac{\,c_{q,r}}{\,t^{\,\frac{3}{2}(\frac{1}{r}-\frac{1}{q})-1}}\|\nabla^2 f\|_{r}\,.
\end{equation}
For the r.h.s.~of (\ref{almostf}) to stay bounded in time, we need $1\leqslant r\leqslant\frac{3q}{2q+3}$, which requires $q \geqslant 3$. This completes the proof of the proposition.
\end{proof}

\medskip

\appendix

\section{Pointwise bounds on the two-body wave function}
\label{sec:t-x-boundedness}

In this section we investigate the boundedness properties in time and in space of the evolution of the two-body wave function $\varphi^{\otimes 2}$ when the two particles are coupled by the interaction $V_N(x_1-x_2)$, that is, where the dynamics is generated by the Hamiltonian
\begin{equation*}
\fh_N^{(1,2)} \;=\; -\Delta_1-\Delta_2+V_N(x_1-x_2)
\end{equation*}
acting on $L^2(\mathbb{R}^3\times\mathbb{R}^3,\rd x_1\rd x_2)$, as defined in (\ref{hN12}). In spirit, these bounds are similar to the ones proved in Proposition~\ref{prop-Omegaomegapsilambda-dispersive} (in particular in (\ref{eq-Omegaomegapsilambda-dispersive}) for $s= \infty$; the initial data however, is different).

\medskip

\begin{proposition}\label{prop:t-x-boundedness}
Let $V$ be a non-negative, smooth, spherically symmetric, and compactly supported potential. Let $V_N(x)=N^2V(Nx)$ and $\fh_N^{(1,2)} = -\Delta_1-\Delta_2+V_N(x_1-x_2)$.
Let $\psi_t = e^{-it\fh_N^{(1,2)}}\!\!\varphi^{\otimes 2}$, for some $\ph \in L^2 (\bR^3)$. Then, for every $\alpha >3$, there exists $C>0$ such that
\begin{equation}\label{eq:t-x-boundedness}
\|\psi_t\|_\infty \; \leqslant \; C\,\|\varphi\|^2_{4,\infty,\alpha}\,\log N\,
\end{equation}
where
\begin{equation}
\|\varphi\|_{4,\infty,\alpha} \; = \; \sum_{m=0}^4\big\|\langle x \rangle^\alpha\nabla^m\varphi\big\|^2_{\infty}.
\end{equation}
\end{proposition}

\medskip

\begin{proof}
Let $\wt\psi (\eta, x) = \psi(\eta+x/2 , \eta -x/2)$, and \[ \wt \psi_t (\eta,x) = \psi_t (\eta+x/2,\eta-x/2) = \left(e^{-i\Delta_{\eta}/2 t} e^{-i\fh_N t} \wt \psi \right) (\eta,x) \]
where the operator $\fh_N = -2 \Delta + V_N (x_1 -x_2)$ only acts on the relative variable $x$. Then
\begin{equation}\label{firstestimate}
\begin{split}
\|\psi_t \|^2_{L^{\infty} (\bR^6, \rd x_1 \, \rd x_2)} \; = \; \;\|\wt \psi_t \|^2_{L^{\infty} (\bR^6, \rd \eta \, \rd x)} \;  & \leqslant \;C\, \sup_{x\in\mathbb{R}^3}\, \|\widetilde\psi_t (\cdot,x)\|^2_{H^2(\mathbb{R}^3,\rd \eta)} \\
& \; \leqslant \;C \sum_{m=0}^2\,\sup_{x\in\mathbb{R}^3}\,\inthree\rd \eta\, \big|\nabla_{\!\eta}^m \wt \psi_t (\eta,x)\big|^2 \\
& \; \leqslant \;C\sum_{m=0}^2\,\sup_{x\in\mathbb{R}^3}\,\inthree\!\rd \eta\,\big| \big(e^{-i\fh_N t} \nabla_{\!\eta}^m \wt \psi \; \big) (\eta,x) \big|^2  \\
& \; \leqslant \;C\sum_{m=0}^2\,\inthree\!\rd \eta\,\big\|\,\big( e^{-i\fh_N t} \nabla_{\!\eta}^m \wt \psi \; \big) (\eta,\cdot) \big\|^2_{L^{\infty}(\mathbb{R}^3,\rd x)} \;.
\end{split}
\end{equation}
It is useful to switch to the macroscopic coordinate $X = Nx$. With the short hand notation
\begin{eqnarray}
\psi_{\eta,N}^{(m)}(X)\! & := &  \psi_\eta^{(m)} \Big(\frac{X}{N}\Big)\, ,  \qquad \text{with }\label{psimN}\\
\psi_\eta^{(m)}(x)\! & := & \! \nabla_{\!\eta}^m \wt \psi (\eta,x) = \! \nabla_{\!\eta}^m  \bigg(\!\varphi \Big(\eta+\frac{x}{2}\Big)\varphi\Big(\eta-\frac{x}{2}\Big)\!\bigg) \, ,
\end{eqnarray}
Eq. (\ref{firstestimate}) reads
\begin{equation}
\|\psi_t \|^2_{L^{\infty} (\bR^6, \rd x_1 \, \rd x_2)}
\; \leqslant \; C\sum_{m=0}^2\,\inthree\!\rd \eta\,\big\| \, e^{-i\fh N^2 t}  \psi^{(m)}_{\eta,N} \big\|^2_{L^{\infty}(\mathbb{R}^3,\rd X)} \; ,
\end{equation}
because $\fh_N = N^2 (-2 \Delta_X + V (X) )$. To bound the r.h.s. of the last equation, we use the modified dispersive estimate (\ref{eq:dispersive_estimate}) with $s=\infty$. We obtain
\begin{equation}\label{disp-m-eps}
\begin{split}
\big\|\, e^{-i \fh N^2 t} &\psi_{N,\eta}^{(m)} \big\|^2_{L^{\infty} (\bR^3, \rd X)} \leq  C\,\Big( \| \Omega^* \psi_{N,\eta}^{(m)} \|^2_\infty + \| \nabla_X \Omega^* \psi_{N,\eta}^{(m)} \|^2_{3} + \| \nabla^2_X \Omega^* \psi^{(m)}_{N,\eta}  \|^2_{\frac{3}{2}} \Big)\, ,
\end{split}
\end{equation}
for every $m=0,1,2$, where $\Omega$ is the wave operator associated with $-\Delta + \frac{1}{2} V (x)=\fh/2$, as defined in Proposition \ref{prop:waveop}.

\medskip

To estimates the terms on the r.h.s. of (\ref{disp-m-eps}) we use the same strategy followed in the proof of Proposition \ref{prop-Omegaomegapsilambda-dispersive}. However, since the initial data considered here (the wave function $\psi^{(m)}_{N,\eta}$) is different from the one considered in Proposition \ref{prop-Omegaomegapsilambda-dispersive}, the bounds for the r.h.s. of (\ref{disp-m-eps}) cannot be directly inferred from the analogous estimates for the terms of the r.h.s.~of (\ref{omegachi_dispersive}). 

\medskip

By Yajima's bound (\ref{yb}), we have
\begin{equation}\label{m1}
\|\,\Omega^*\psi_{\eta,N}^{(m)}\|_{\infty}\;\leqslant\;C\,\|\,\psi_{N,\eta}^{(m)}\|_{\infty} = C \, \| \psi_{\eta}^{(m)} \|_{\infty}.
\end{equation}

\medskip

Analogously to the treatment of $\|\nabla\Omega^*\omega\psi_\Lambda\|_{v}$
after (\ref{treatment_of_nablaomegapsiL}), we obtain
\begin{equation}
\|\nabla_{\!\!X}\Omega^*\psi_{\eta,N}^{(m)}\|_{3}\; \leqslant \; \frac{C}{\varepsilon(1-\varepsilon)^{1/3}}\, \Big( \|\,\fh\,\psi_{\eta,N}^{(m)}  \,\|_{\frac{3}{2+\varepsilon}} +
\|\,\fh\,\psi_{\eta,N}^{(m)}  \,\|_{3}\Big)\qquad\qquad\varepsilon\in(0,1)\,.
\end{equation}
Since
\begin{eqnarray}
\|\,\Delta_X\,\psi_{\eta,N}^{(m)} \|_{\frac{3}{2+\varepsilon}} & = &  N^{\varepsilon}\|\,\Delta_x\,\psi_{\eta}^{(m)}\|_{\frac{3}{2+\varepsilon}} \\
\|\,\Delta_X\,\psi_{\eta,N}^{(m)}\|_{3} &  = & \frac{1}{\,N}\|\,\Delta_x\,\psi_{\eta}^{(m)}  \|_{3}\,,
\end{eqnarray}
one has
\begin{equation}\label{m2}
\begin{split}
\|\nabla_{\!\!X}\Omega^*\psi_{\eta,N}^{(m)}\|_{3} \; & \leqslant \; \frac{C}{\varepsilon(1-\varepsilon)^{\frac{1}{3}}} \,\Big(\,N^{\eps} \, \| \Delta_x\psi_{\eta}^{(m)}\|_{\frac{3}{2+\varepsilon}} + \frac{1}{\,N}\|\,\Delta_x\psi_{\eta}^{(m)}  \|_{3}+\|\psi_{\eta}^{(m)}\|_\infty \Big)\,.
\end{split}
\end{equation}

\medskip

Similarly to the treatment of $\|\nabla^2\Omega^*\omega\psi_\Lambda\|_{z}$ in (\ref{nabla2OmegaomegachiL}), we find
\begin{equation}\label{m3}
\begin{split}
\|\nabla^2_{\!\!X}\Omega^*\psi_{\eta,N}^{(m)}\|_{\frac{3}{2}} \; & \leqslant \; C\, \Big(\,\|\Delta_{X}\psi_{\eta,N}^{(m)}  \|_{\frac{3}{2}} +\|V\psi_{\eta,N}^{(m)}  \|_{\frac{3}{2}}\Big)\; \leqslant \; C\,\Big(\, \|\Delta_x\psi_{\eta}^{(m)}\|_{\frac{3}{2}}+\|\psi_{\eta}^{(m)}\|_\infty\Big)\,.
\end{split}
\end{equation}

\medskip

Using (\ref{m1}), (\ref{m2}), and (\ref{m3}) to bound the r.h.s.~of (\ref{disp-m-eps}), we find
\begin{equation}
\begin{split}
\big\|\, e^{-i \fh N^2 t} \psi_{N,\eta}^{(m)} &\big\|^2_{L^{\infty} (\bR^3, \rd X)} \\ &\leq C \left(\frac{\;N^{\varepsilon}}{\varepsilon(1-\varepsilon)^{\frac{1}{3}}}\right)^2 \, \left(  \, \|\,\Delta_x\psi_{\eta}^{(m)}\|^2_{\frac{3}{2+\varepsilon}} + \|\Delta_x\psi_{\eta}^{(m)}\|^2_{\frac{3}{2}}+\frac{1}{\,N}\|\,\Delta_x\psi_{\eta}^{(m)}  \|^2_{3} + \|\psi_{\eta}^{(m)}\|^2_\infty \right).
\end{split}\end{equation}
Now it remains to integrate the r.h.s. of the last equation over $\eta \in \bR^3$. By the assumptions on $\ph$, it follows that
\begin{equation}
|\nabla^\nu\varphi(x)|\;  < \; \|\varphi\|_{4,\infty,\alpha}\frac{1}{\,\langle x\rangle^\alpha} 
\end{equation}
for any $\nu=0,\dots,4$ and some $\alpha>3$. This means that for any $p\in[1,+\infty)$ and $n,m\in\{0,1,2\}$ one has
\begin{equation}
\begin{split}
\big\|\nabla_{\!x}^n\psi_{\eta}^{(m)}\big\|_p\; & = \; \bigg(\inthree\rd x\, \Big|\nabla_{\!x}^n\nabla_\eta^m\Big(\varphi\Big(\eta+\frac{x}{2}\Big)\varphi\Big(\eta-\frac{x}{2}\Big)\Big)  \Big|^p\bigg)^{\!\frac{1}{p}} \\
& \leqslant\; C\sum_{\mu=0}^{m+n}\bigg(\inthree\rd x\,\Big|(\nabla^\mu\varphi)\Big(\eta+\frac{x}{2}\Big)\Big|^p \Big|(\nabla^{m+n-\mu}\varphi)\Big(\eta-\frac{x}{2}\Big)\Big|^p\bigg)^{\!\frac{1}{p}} \\
& \leqslant \; C\,\|\varphi\|^2_{4,\infty,(\alpha)} \bigg(\inthree\rd x\,\frac{1}{\,\langle \eta+\frac{x}{2}\rangle^{\alpha p}} \frac{1}{\,\langle \eta-\frac{x}{2}\rangle^{\alpha p}}  \bigg)^{\!\frac{1}{p}} 
\end{split}
\end{equation}
and the above integral can be estimated as
\begin{equation}
\inthree\rd x\, \frac{1}{\,\langle \eta+\frac{x}{2}\rangle^{\alpha p}} \frac{1}{\,\langle \eta-\frac{x}{2}\rangle^{\alpha p}}\; \leqslant \; \frac{C}{\;\langle \eta \rangle^{\alpha p}}
\end{equation}
as long as $\alpha p>3$, which is always the case due to the assumption $\alpha>3$. Then
\begin{equation}\label{superest}
\big\|\nabla_{\!x}^n\psi_{\eta}^{(m)}\big\|^2_p\;\leqslant \; C\, \frac{\,\|\varphi\|^4_{4,\infty,\alpha}}{\;\langle \eta \rangle^{2\alpha}} \qquad\qquad n,m\in\{0,1,2\}
\end{equation}
independently of $p\in[1,+\infty]$.

\medskip

Due to (\ref{superest}), when plugging (\ref{m1}), (\ref{m2}), and (\ref{m3}) into the r.h.s.~of estimate (\ref{disp-m-eps}), and the latter into the r.h.s.~of (\ref{firstestimate}), integrability in $\eta$ is guaranteed by the assumption $\alpha>3$ and one gets
\begin{equation}
\|\psi_t\|_\infty
\leqslant \|\varphi\|^2_{4,\infty,\alpha}\,\frac{\:N^{\varepsilon}}{\,\varepsilon(1-\varepsilon)^{\frac{1}{3}}} \qquad\qquad\varepsilon\in(0,1)\,.
\end{equation}
After choosing $\varepsilon=(\log N)^{-1}$, estimate (\ref{eq:t-x-boundedness}) is proved.
\end{proof}

\section{Estimates for the energy of a factorized data}

In the following lemma we prove that the expectation of the Hamiltonian $H_N$ in a factorized state $\ph^{\otimes N}$, for some $\ph \in H^2 (\bR^3)$, is of the order $N$. This estimate is used in Section \ref{sec:Nbodyproblem}. Moreover, we also show that the expectation of $H_N^2$ (in the state $\ph^{\otimes N}$) is of the order $N^3$.

\begin{lemma}\label{prodphi-N3}
Let $H_N$ be defined as in (\ref{N-Hamiltonian}) and let $\varphi\in H^2(\mathbb{R}^3)$. Then
\begin{align}
\lim_{N\to\infty}\frac{1}{\,N}\big\langle\,\varphi^{\otimes N}\!,H_N\,\varphi^{\otimes N}\big\rangle & = \|\nabla\varphi\|_2^2+\!\!\begin{array}{l}\frac{1}{2}\end{array}\!\!\|V\|_1\|\varphi\|_4^4 & \forall\varphi\in H^1(\mathbb{R}^3)\,,\qquad & \label{expectation-H}\\
\lim_{N\to\infty}\frac{1}{\,N^3}\big\langle\,\varphi^{\otimes N}\!,H_N^2\,\varphi^{\otimes N}\big\rangle & =  \!\!\begin{array}{l}\frac{1}{2}\end{array}\!\!\|\varphi\|_4^4 \|V\|_2^2 & \forall\varphi\in H^2(\mathbb{R}^3)\,.\qquad & \label{expectation-H2}
\end{align}
\end{lemma}

\medskip

\begin{proof}
We will prove equation (\ref{expectation-H2}), the proof of
(\ref{expectation-H}) is similar but simpler, so we will omit it.
Using the permutation symmetry of $\ph^{\otimes N}$, we find that
\begin{equation}\label{eq:bd0}
\begin{split}
\frac{1}{N^3} \Big| \langle \ph^{\otimes N} , H_N^2 \ph^{\otimes N} &\rangle - \frac{N(N-1)}{2} \langle \ph^{\otimes N}, V^2_N (x_1 -x_2) \ph^{\otimes N} \rangle \Big| \\ \leq \; &N^{-1} \langle \ph^{\otimes N}, \Delta_1 \Delta_2 \ph^{\otimes N} \rangle + N^{-2} \langle \ph^{\otimes N} , \Delta_1^2 \ph^{\otimes N} \rangle \\ &+  \langle \ph^{\otimes N} ,(- \Delta_1) V_N (x_2 -x_3) \ph^{\otimes N} \rangle + \left| \langle \ph^{\otimes N}, \Delta_1 V_N (x_1 -x_2) \ph^{\otimes N} \rangle \right| \\ &+ N \langle \ph^{\otimes N}, V_N (x_1 -x_2) V_N (x_3-x_4) \ph^{\otimes N} \rangle \\& + \langle \ph^{\otimes N}, V_N (x_1 -x_2) V_N (x_2 -x_3) \ph^{\otimes N} \rangle\,.
\end{split}
\end{equation}
Since $\ph \in H^2 (\bR^3)$ and $\| \ph \|=1$, it follows that
\begin{equation}\label{eq:bd1} \langle \ph^{\otimes N}, \Delta_1 \Delta_2 \ph^{\otimes N} \rangle \leq \| \ph \|_{H^1}^4 < \infty \qquad \text{and} \qquad \langle \ph^{\otimes N} , \Delta^2_1 \ph^{\otimes N} \rangle \leq \| \ph \|^2_{H^2} < \infty \,. \end{equation}
Moreover, using Lemma \ref{lm:sob}, we have \begin{equation*} \langle \ph^{\otimes N} , (-\Delta_1) V_N (x_2 -x_3) \ph^{\otimes N} \rangle \leq  C \, \| V_N \|_1 \, \| \ph \|^6_{H^1}  \leq C N^{-1} \, \| V \|_1 \, \| \ph \|^6_{H^1}  \end{equation*} and, with a Schwarz inequality and Sobolev embedding,
\begin{equation*}
\begin{split}
\Big |\langle \ph^{\otimes N}, \Delta_1 V_N (x_1 -x_2) \ph^{\otimes N} \rangle \Big| \leq \; &\langle \ph^{\otimes N}, \Delta_1 V_N (x_1 -x_2) \Delta_1 \ph^{\otimes N} \rangle^{1/2} \, \langle \ph^{\otimes N}, V_N (x_1 -x_2) \ph^{\otimes N} \rangle^{1/2} \\
\leq \; & C N^{-1} \, \| V \|_1 \, \| \ph \|_{H^2}^4 \,.
\end{split}
\end{equation*}
Finally, using again Lemma \ref{lm:sob}, we observe that
\begin{equation*}
\langle \ph^{\otimes N}, V_N (x_1 -x_2) V_N (x_3-x_4) \ph^{\otimes N} \rangle \leq C \, \| \ph \|^8_{H^1} \| V_N \|_1^2 \leq C N^{-2} \| V \|^2_1 \, \| \ph \|^8_{H^1} \end{equation*} and that
\begin{equation}\label{eq:bd6}
\langle \ph^{\otimes N}, V_N (x_1 -x_2) V_N (x_2 -x_3) \ph^{\otimes N} \rangle \leq C \, \| \ph \|^4_{H^1} \| V \|_{3/2}^2 \, . \end{equation}
Inserting all these bounds in the r.h.s. of (\ref{eq:bd0}), it follows that
\begin{equation}
\lim_{N \to \infty} \frac{1}{N^3} \Big| \langle \ph^{\otimes N} , H_N^2 \ph^{\otimes N} \rangle - \frac{N(N-1)}{2} \langle \ph^{\otimes N},  V^2_N (x_1 -x_2) \ph^{\otimes N} \rangle \Big| = 0\, .
\end{equation}
Eq. (\ref{expectation-H2}) now follows because
\[ \frac{1}{2N} \langle \ph^{\otimes N}, V_N^2 (x_1 - x_2) \ph^{\otimes N} \rangle = \frac{1}{2} \int \rd x_1 \rd x_2 \, N^3 V^2 (N (x_1 -x_2)) |\ph (x_1)|^2 |\ph (x_2)|^2 \to \frac{\| V^2 \|_1}{2} \| \ph \|_4^4 \] as $N \to \infty$ (the convergence follows by a Poincar{\'e} inequality, since $\ph \in H^2 (\bR^3)$).
\end{proof}

\begin{lemma}[Sobolev-type inequalities]\label{lm:sob}
Let $\psi \in L^2 (\bR^6, \rd x_1 \rd x_2)$. If $V \in L^{3/2} (\bR^3)$, we have
\begin{equation}\label{eq:3/2}
\left| \langle \psi, V(x_1 -x_2) \psi \rangle \right| \leq C \| V \|_{3/2} \, \langle \psi, (1-\Delta_1) \psi \rangle \, .
\end{equation}
If $V \in L^1 (\bR^3)$, then \begin{equation}\label{eq:Vsob} \left| \langle \psi, V(x_1 -x_2) \psi \rangle \right| \leq C \| V \|_1 \, \langle \psi, (1-\Delta_1) (1-\Delta_2) \psi \rangle \end{equation}
\end{lemma}
The first bound follows from a H\"older inequality followed by a standard Sobolev inequality (in the variable $x_1$, with fixed $x_2$). A proof of the second bound can be found, for example, in \cite{EY}[Lemma 5.3].

\section{Properties of the one-body scattering solution $1-\omega(x)$}\label{app1bodyscatt}

\begin{lemma}\label{lemmaomega}
Let $V:\mathbb{R}^3\to\mathbb{R}$ be non negative, smooth, spherically symmetric, compactly supported and with scattering length $a$. Let
$1-\omega(x)$ be the solution of
\begin{equation}
\big(\!-\Delta+\!\!\!\!\begin{array}{l}\frac{1}{2}\end{array}\!\!\! V\big)(1-\omega) \; = \; 0 \qquad\qquad\textrm{with}\quad \omega(x)\to 0\quad\textrm{as}\quad |x|\to\infty\,.
\end{equation}
Then there exist constants $C_m,\widetilde C$ depending on the potential $V$ such that
\begin{align}
|\omega(x)| &\leqslant C_0\,\frac{a}{|x|}                 & |\omega(x)| & \leqslant \widetilde C \qquad (\widetilde C<1) \label{bound-on-omega}\\
|\nabla^m\omega(x)| &\leqslant C_m\,\frac{a}{\:|x|^{m+1}}                 & |\nabla^m\omega(x)| & \leqslant C_m \label{bound-on-nablaomega}\
\end{align}
for every nonnegative integer $m$.
As a consequence, for every $p\in[1,\infty]$ and every nonnegative integer $m$ such that $p(m+1)>3$, one has
\begin{equation}\label{nablaomeganorms}
\|\nabla^m\omega\|_p \; < \; C_{m,p} \; < \; \infty\,.
\end{equation}
\end{lemma}

\begin{proof}
Inequalities (\ref{bound-on-omega}) and (\ref{bound-on-nablaomega}) follow immediately from the fact that, out of the support of the potential $V$ one has
\begin{equation}
\omega(x) \; = \; \frac{a}{|x|} \qquad\qquad |x|> R\,,
\end{equation}
while, inside the support, $\nabla^m\omega$ is bounded by elliptic regularity and compactness for any nonnegative $m$. The fact that the constant $\widetilde C$ in (\ref{bound-on-omega}) is strictly smaller than 1 is proved in Lemma B.1 of \cite{ESY-2006}.
\end{proof}

By scaling, one immediately has the following.

\begin{corollary}\label{lemmaomegaN}
Let $V$ be as in Lemma \ref{lemmaomega}. Let $V_N(x)=N^2V(Nx)$ and let $1-\omega_N$ be the corresponding solution of the zero-energy scattering equation. Then
\begin{equation}
\omega_N(x) \; = \; \omega(Nx)
\end{equation}
whence
\begin{equation}
\|\nabla^m\omega_N\|_p \; = \; N^{m-\frac{3}{p}} \|\nabla^m\omega\|_p
\end{equation}
under the same condition for the validity of (\ref{nablaomeganorms}).
\end{corollary}

\section{Properties of the wave operator $\Omega$}\label{app:waveop}

We denote by $\Omega$ the wave operator associated with the one-particle Hamiltonian $h:=-\Delta+\frac{1}{2}V$, that is, $h=\frac{1}{2}\fh$ in our previous notation. Its existence and most important properties are stated in the following proposition.

\medskip

\begin{proposition} \label{prop:waveop}
Suppose $V \geq 0$, with $V \in L^1 (\bR^3)$. Then:
\begin{itemize}
\item[i)] ({\it Existence of the wave operator}). The limit \[ \Omega = s-\lim_{t\to \infty} e^{i h t} e^{i\Delta t} \] exists.
\item[ii)] ({\it Completeness of the wave operator}). $\Omega$ is a unitary operator on $L^2 (\bR^3)$ with \[ \Omega^* = \Omega^{-1} =  s-\lim_{t\to \infty} e^{-i\Delta t} e^{-i h t} \]
\item[iii)] ({\it Intertwining relations}). On $D(h) = D (-\Delta)$, we have
\begin{equation}\label{eq:intertw}
\Omega^* h\, \Omega  = - \Delta
\end{equation}
\item[iv)] ({\it Yajima's bounds}). Suppose moreover that $V(x)\leq C \langle x \rangle^{-\sigma}$, for some $\sigma >5$. Then, for every $1 \leq p \leq \infty$, $\Omega$ and $\Omega^*$ map $L^p (\bR^3)$ into $L^p (\bR^3)$, that is,
\begin{equation}\label{yb}
\| \Omega \|_{L^p \to L^p} < \infty \qquad \text{for all } \quad 1 \leqslant p \leqslant \infty\,.
\end{equation}
If moreover $V \in C^k (\bR^3)$, we have \[ \| \Omega \|_{W^{m,p} \to W^{m,p}} < \infty \] for all $m \leq k$.
\end{itemize}
\end{proposition}

\medskip

\begin{proof}
The proof of i), ii), and iii) can be found in \cite{rs3}. Part iv) is proved in \cite{yajima-93,yajima-95}.
\end{proof}

\section{Correlation structure and Gross-Pitaevskii equation}
\label{app:GPa}

As remarked in the introduction, the correlation structure developed by the solution to the $N$-particle Schr\"odinger equation $\Psi_{N,t} = e^{-iH_N t} \Psi_N$ (with $H_N$ defined as in (\ref{HN})) for initial states exhibiting complete Bose-Einstein condensation plays a very important role in the derivation of the Gross-Pitaevskii equation (\ref{tGP}); more precisely, the emergence of the scattering length in the coupling constant in front of the nonlinearity is a consequence of the presence of the short scale correlation structure. The goal of this appendix is to explain this connection in some more details (see also Section 3 in \cite{ESY-2006}).

\medskip

{F}rom the Schr\"odinger equation (\ref{HNev}), it is simple to obtain an evolution equation for the one-particle density $\gamma^{(1)}_{N,t}$:
\begin{equation}\label{eq:BB1}
\begin{split}
i\partial_t \gamma^{(1)}_{N,t} (x_1; x'_1) = \; &\left( -\Delta_{x_1} + \Delta_{x'_1}\right)  \gamma^{(1)}_{N,t} (x_1 ; x'_1) \\ &+ (N-1) \int \rd x_2 \left( V_N (x_1 -x_2) - V_N (x'_1 -x_2) \right) \gamma^{(2)}_{N,t} (x_1,x_2; x'_1, x_2) \, .
\end{split}
\end{equation}
This is not a closed equation for $\gamma^{(1)}_{N,t}$ because it also depends on the two-particle density $\gamma_{N,t}^{(2)}$ associated with $\Psi_{N,t}$ (the two-particle density is defined similarly to (\ref{eq:1part}), integrating however only over the last $(N-2)$ particles); actually, (\ref{eq:BB1}) is the first equation of a hierarchy of $N$ coupled equations, known as the BBGKY hierarchy, for the marginal densities of $\Psi_{N,t}$. {F}rom $\gamma^{(1)}_{N,t} \to |\ph_t \rangle \langle \ph_t|$ as $N \to \infty$, it follows that $\gamma^{(2)}_{N,t} \to |\ph_t \rangle \langle \ph_t|^{\otimes 2}$. If we replace, in (\ref{eq:BB1}), the densities $\gamma^{(1)}_{N,t}$ and $\gamma^{(2)}_{N,t}$ by these limit points, and if we replace $(N-1) V_N (x) \simeq N^3 V (Nx)$ by its (formal) limit $b \delta (x)$ with $b = \int V$, we obtain a closed equation for the condensate wave function $\ph_t$, which has the same form as the Gross-Pitaevskii equation (\ref{tGP}), but with a coupling constant in front of the nonlinearity given by $b$ instead of $8\pi a$.  The reason why this naive argument leads to the wrong coupling constant is that the two-particle density $\gamma^{(2)}_{N,t}$ contains a short scale correlation structure (inherited by the $N$-particle wave function $\Psi_{N,t}$) which varies on exactly the same length scale $N^{-1}$ characterizing the interaction potential. Describing correlations by the solution to the zero energy scattering equation, we can approximate, for large but finite $N$, the two-particle density $\gamma^{(2)}_{N,t}$ by \[ \gamma^{(2)}_{N,t} (x_1, x_2; x'_1,x'_2) \simeq (1-\omega_N (x_1 -x_2)) (1-\omega_N (x'_1 - x'_2)) \ph_t (x_1) \ph_t (x_2) \overline{\ph}_t (x'_1) \overline{\ph}_t (x'_2)\,. \] Inserting this ansatz in the second term on the r.h.s. of (\ref{eq:BB1}), and using (\ref{scattlength}), we obtain the correct Gross-Pitaevskii equation for $\ph_t$. The emergence of the scattering length in the Gross-Pitaevskii equation is therefore a consequence of the singular correlation structure developed by $\Psi_{N,t}$ (and then inherited by the two particle marginal density).

\bibliographystyle{siam}

\end{document}